\documentclass[a4paper,11pt]{article}

\usepackage{graphicx}
\usepackage{fullpage}
\usepackage{libertine}
\usepackage{color}

\usepackage[ruled,vlined]{algorithm2e}
\SetArgSty{textrm}

\usepackage{amsmath,amsfonts,amssymb,amsthm}

\usepackage[breaklinks=true]{hyperref}
\usepackage[svgnames]{xcolor}
\usepackage[capitalise,nameinlink]{cleveref}
\hypersetup{colorlinks={true},linkcolor={DarkBlue},citecolor=[named]{DarkGreen}}

\usepackage{natbib}

\usepackage{subcaption}

\usepackage{tikz}  
\usetikzlibrary{patterns,snakes}
\usetikzlibrary{decorations.shapes}
\tikzstyle{overbrace text style}=[font=\tiny, above, pos=.5, yshift=5pt]
\tikzstyle{overbrace style}=[decorate,decoration={brace,raise=5pt,amplitude=3pt}]
\usetikzlibrary{shapes.geometric}

\newtheorem{theorem}{Theorem}[section]
\newtheorem{corollary}[theorem]{Corollary}

\newtheorem{lemma}[theorem]{Lemma}

\theoremstyle{definition}
\newtheorem*{comment*}{Comment}

\newcommand{\SW}{\text{SW}}
\newcommand{\CE}{\text{CE}}

\newcommand{\NE}{\text{NE}}
\newcommand{\PoA}{\text{PoA}}
\newcommand{\PoS}{\text{PoS}}
\newcommand{\bv}{\mathbf{v}}
\newcommand{\Tau}{\text{TE}}
\newcommand{\po}{$\mathbf{(P_0)}$}

\newcommand{\pOne}{$\mathbf{(P_1)}$}

\title{\bf Variety-Seeking Jump Games on Graphs}

\usepackage{authblk}
\author[1]{Lata Narayanan}
\author[1]{Jaroslav Opatrny}
\author[1]{Shanmukha Tummala}
\author[2]{Alexandros A. Voudouris}

\affil[1]{Department of Computer Science and Software Engineering, Concordia University}
\affil[2]{School of Computer Science and Electronic Engineering, University of Essex}

\date{}

\begin{document}

\maketitle

\begin{abstract}
We consider a class of jump games in which agents of different types occupy the nodes of a graph aiming to maximize the {\em variety of types} in their neighborhood. In particular, each agent derives a utility equal to the {\em number of types different from its own} in its neighborhood. We show that the jump game induced by the strategic behavior of the agents (who aim to maximize their utility) may in general have improving response cycles, but is a potential game under {\em any} of the following four conditions: there are only two types of agents; or exactly one empty node;  or the graph is of degree at most 2; or the graph is 3-regular and there are two empty nodes. Additionally, we show that on trees, cylinder graphs, and tori, there is always an equilibrium. Finally, we show tight bounds on the price of anarchy with respect to two different measures of diversity: the social welfare (the total utility of the agents) and the number of colorful edges (that connect agents of different types).
\end{abstract}

\section{Introduction}
While analyzing residential segregation, \citet{schelling1969models,schelling1971dynamic} imagined a simple scenario in which agents of two different types are randomly assigned to the nodes of a graph representing a city. The agents are allowed to randomly {\em jump} to other available locations or {\em swap} locations with other agents whenever this increases the fraction of same-type neighbors they have, up to a threshold. Schelling experimentally showed that, in most cases, this random behavior of the agents leads to segregated neighborhoods consisting only of agents of one type. His work inspired researchers in many different disciplines to further study this model and generalize it to capture more complicated dynamics between agents of different types.  Recent work within the multi-agent literature has considered game-theoretic variants of Schelling's model in which agents act {\em selfishly} rather than randomly by aiming to maximize a {\em utility} function. This strategic behavior of the agents then defines a {\em game} between them and the objective is to analyze whether {\em equilibria} (stable assignments where agents simultaneously achieve the maximum possible utility they can) exist and what properties (related to segregation) they have. 

With few exceptions, most of the utility functions that have been proposed and studied over the years can be described as {\em similarity-seeking} or {\em homophilic} in the sense that agents prefer to be close to other agents of the same type~\citep{ijcai2022p22,bullinger2021welfare,chauhan2018schelling,echzell2019convergence,schelling-journal,kanellopoulos2020modified}. Such a behavior is well-justified in scenarios where agents aim to form groups that share similar interests or require similar skill sets to complete tasks. On the other hand, however, there many other important applications where the desideratum is {\em diversity}. For example, data from the General Social Survey \citep{smith2019general} (conducted in the US since 1950) show that a steadily increasing percentage of people prefer to live in diverse neighborhoods. In addition, many governments, businesses, and other institutions are actively promoting increased diversity as being beneficial to both societal harmony and efficiency. 

Motivated by applications like those mentioned above, some recent papers have considered utility functions which aim to model diversity as being beneficial for the agents. In particular, \citet{ijcai2022p12} and \citet{friedrich2023single} studied games with two types of agents and a single-peaked utility function which increases monotonically with the fraction of same-type neighbors in the interval $[0, \Lambda]$ for some $\Lambda \in (0,1)$, and then decreases monotonically. A different model was proposed by \citet{kanellopoulos2023tolerance} who focused on a utility function that assigns different weights to different types according to their position on a line, indicating different preferences over the types. More recently, \citet{Narayanan2023diversity} focused on jump games where the utility of an agent is the fraction of its {\em different-type} neighbors; essentially, this function is the complement of the one implied by Schelling's original model and used in a plethora of subsequent works. 

\subsection{Our Contribution}
An important aspect of the models studied in the aforementioned papers is that the agents are mainly concerned with the type difference between themselves and their neighbors, and not the difference or diversity {\em among} their neighbors. For example, according to the utility function considered by \citet{Narayanan2023diversity}, an agent is fully satisfied even if all its neighbors are of one type, as long as this type is different than its own. Whether this is a truly diverse neighborhood is up to debate. In this paper, we study jump games where agents are seeking {\em variety}: An agent's utility is defined as the {\em number of types different from its own} among its neighbors. Such a utility function was very recently studied for swap games by \citet{LNO25}. With this utility function, the maximum possible utility is the maximum between the number of types and the degree of the underlying graph. 

To be more specific, we consider jump games with $n$ agents that are partitioned into $k \geq 2$ different types and occupy the nodes of a graph $G=(V, E)$, where $|V| < n$. Each agent aims to maximize the number of different types in its neighborhood. We first show that such games may have an {\em improving response cycle}, even when $G$ is $3$-regular and $|V|=n+3$. This means that there is an initial assignment of the agents to the nodes of the graph so that if the agents jump one by one to empty nodes of the graph, they will eventually cycle back to that initial assignment; in other words, starting from such an assignment, an equilibrium cannot be reached. In spite of this impossibility, we show that the game is {\em potential} (that is, the Nash dynamics converges to an equilibrium) under {\em any} of the following conditions: $k=2$; $|V|=n+1$; the graph is of degree at most $2$; the graph is $3$-regular and $|V|=n+2$. Additionally, we show that there is always an equilibrium when the graph is a tree, or a cylinder, or a torus by carefully constructing one for these cases. 

We next switch to analyzing the quality of equilibria under two different objectives: The {\em social welfare} (defined as the total utility of the agents) and the {\em number of colorful edges} (defined as the number of edges between agents of different types). Both objectives are different measures of diversity and have been considered before by \cite{LNO25} for the case of swap games. We show tight bounds on the {\em price of anarchy}, which quantifies the loss in the social welfare or the number of colorful edges in the worst equilibrium. In particular, we show that the price of anarchy with respect to social welfare is $\Theta(n)$ for general games, and $\Theta(k)$ for symmetric types (that is, when all types are of the same cardinality). For colorful edges, we show that the price of anarchy is $\Theta(n)$  even for symmetric types, and $\Theta(\delta)$ when the graph is $\delta$-regular and the types are symmetric. In addition, for both objectives, we provide lower bounds on the price of stability, showing that there are instances in which the optimal assignment is not necessarily an equilibrium. 

\subsection{Related Work}
As already mentioned, many generalizations and variations of Schelling's random model have been studied via agent-based simulations in many different disciplines, including sociology \citep{Clark2008understanding}, economics \citep{Pancs2007spatial,Zhang2004residential}, and physics \citep{Vinkovic2006schelling}. In computer science, most of the work has focused on probabilistic analyses of Schelling's model~\citep{Barmpalias2014digital,Bhakta2014clustering,Brandt2012one,Immorlica2017exponential,Blasius2023flip} and also on the computational complexity of assignments with certain efficiency properties~\citep{bullinger2021welfare,deligkas2024parameterized}. 

Our work is related to a recent stream of papers that consider Schelling games induced when the agents do not act randomly but strategically by maximizing a utility function. The majority of papers in this literature study jump and swap games with similarity-seeking utility functions that one way or another depend on the ratio of same-type agents in the neighborhood, and consider questions related to the existence of equilibria, their computational complexity, and their quality as measured by the price of anarchy and the price of stability. Some of the first papers in this area include the works of \citet{chauhan2018schelling} and \citet{echzell2019convergence} who studied the convergence of the best response dynamics to equilibrium assignments, and the work of \citet{schelling-journal} who showed that equilibria may not always exist and, even when they do, they might be hard to compute. Much of the follow-up work provided stronger hardness results~\citep{Kreisel2024equilibria}, studied different utility functions aiming to model other natural agent behaviors~\citep{kanellopoulos2020modified,bilo2022topological}, and made different assumptions about how agents are related to each other~\citep{Bilo2023continuous,chan2020schelling}  

Some recent work on Schelling games has deviated from the standard assumption that agents aim to be close to their own type and have proposed different models in which the agents derive utility from other types of agents as well~\citep{ijcai2022p12,friedrich2023single,kanellopoulos2023tolerance,Narayanan2023diversity,LNO25}. The paper most related to ours is that of \cite{LNO25} who studied swap games with several different diversity-seeking utility functions, including the one we consider here, that is the number of different types in an agent's neighborhood. Among other results, in contrast to our work here where we show that there might exist improving response cycles, they showed that swap games are always potential no matter the structure of the underlying graph. They also showed tight bounds on the price of anarchy for objectives such as the number of colorful edges that we also consider, as well as some other ones.
Most of their bounds indicate that the price of anarchy depends on the degree of the graph and in many case tends to $1$ as the number of agents $n$ or the number of types $k$ grows; in contrast, our price of anarchy bounds suggest a linear dependency on $n$ or $k$ in most cases. 

\section{Preliminaries}
There is a set $N$ of $n\geq 2$ {\em agents} who are partitioned into $k \geq 2$ {\em types} $\mathcal{T} = \{T_1,\ldots,T_k\}$. We denote by $t(i)$ the type of agent $i$; that is, $t(i) = T$ if $i \in T$. We denote by $n_T$ the number of agents of type $T$; if $n_T = n/k$ for each type $T$, then the types are called {\em symmetric}. Each agent $i$ occupies a node of a connected graph $G = (V,E)$ with $|V| > n$. An {\em assignment} $\bv = (v_i)_{i \in N}$ specifies the node $v_i$ that each agent $i$ occupies in $G$, such that $v_i \neq v_j$ for different agents $i$ and $j$. Given an assignment $\bv$, we denote by $N_i(\bv)$ the {\em neighbors} of $i$, which is the set of all agents that occupy nodes adjacent to $v_i$ in $G$. Also, for any  node $v \in V$ and assignment $\bv$, we denote by $T_{v}(\bv)$ the set of agent types located at nodes adjacent to $v$. Given this, the {\em type-count} of node $v$ is $\tau_v(\bv)=|T_{v}(\bv)|$.
Depending on the formation of its neighborhood, each agent $i$ gains a {\em utility} $u_i(\bv)$ equal to the number of types different than $t(i)$ in $i$'s neighborhood.

An assignment $\bv$ is a {\em pure Nash equilibrium} (or, simply, {\em equilibrium}) if no agent can strictly increase its utility by unilaterally jumping to an {\em empty} node in the graph. That is, for every agent $i$ and empty node $v$ (according to $\bv$), $u_i(\bv) \geq u_i(\bv')$, where $\bv'$ is the same as $\bv$ with the only exception that $v_i' = v$. For an instance $I = (N,\mathcal{T},G)$, let $\NE(I)$ be the set of equilibrium assignments; note that this set might be empty.

We are interested in characterizing the classes of graphs for which equilibrium assignments are guaranteed to exist. To do so, we will in many cases show that the best-response Nash dynamics converges to an equilibrium by identifying an {\em ordinal potential function} $\Phi(\bv)$. Such a function has the following property: For any two assignments $\bv$ and $\bv'$ that differ only on the node occupied by an agent $i$, it holds that 
\begin{align*}
    ( u_i(\bv) - u_i(\bv') )\cdot ( \Phi(\bv) - \Phi(\bv') ) > 0.
\end{align*}
Essentially, this property requires that the function is strictly increasing whenever there is an agent with a deviating strategy that leads to strictly larger utility. Note that a strictly decreasing function can be also used as a potential since its monotonicity can be switched by changing its sign. A potential function can be defined only if there is no {\em improving response cycle} (IRC) in the Nash dynamics (that is, a sequence of assignments in which the first and last assignments are the same, and two consecutive assignments differ only on the node occupied by an agent who has larger utility in the latter assignment of the two). 

When equilibrium assignments exist, we are also interested in measuring their quality in terms of achieved diversity. To do this, we consider two objectives functions:
\begin{itemize}
    \item The {\em social welfare} ($\SW$), defined as the total utility of the agents:
    \begin{align*}
        \SW(\bv) = \sum_{i \in N} u_i(\bv).
    \end{align*}
    \item The {\em number of colorful edges ($\CE$)} which is the number of edges whose endpoints are occupied by agents of different types. 
\end{itemize}
For each objective $f\in \{\SW,\CE\}$, we define the {\em price of anarchy} as the worst-case ratio (over a class of instances $\mathcal{I}$) between the maximum possible $f$-value over all assignments and the {\em minimum} $f$-value over all equilibrium assignments:
\begin{align*}
   \PoA_f =  \sup_{I \in \mathcal{I}} \frac{\max_\bv f(\bv)}{\min_{\bv \in \NE(I)} f(\bv)}.
\end{align*}
Similarly, the {\em price of stability} is defined as the worst-case ratio (over a class of instances $\mathcal{I}$) between the maximum possible $f$-value over all assignments and the {\em maximum} $f$-value over all equilibrium assignments:
\begin{align*}
   \PoS_f = \sup_{I \in \mathcal{I}} \frac{\max_\bv f(\bv)}{\max_{\bv \in \NE(I)} f(\bv)}.
\end{align*}
Observe that, be definition, $\PoA_f \geq \PoS_f \geq 1$.  

\section{Existence of Equilibria}
In this section we focus the existence of equilibrium assignments. We first start with an impossibility: There exists a game with an improving response cycle in its Nash dynamics. This implies that the jump game we consider is not always a potential game, and thus the existence of an equilibrium assignment can not always been shown by constructing a potential function. 

\begin{theorem} \label{thm:IRC}
There exists a game with an improving response cycle in the Nash dynamics, even when the graph is $3$-regular and there are three empty nodes. 
\end{theorem}

\begin{proof}
Consider the improving response cycle shown in Figure~\ref{fig:IRC}. In the first assignment in the left column, the red agent occupying node $a$ has utility $1$ (as it only has blue neighbors) and prefers to jump to node $d$ to improve its utility to $2$. 
This leads to the second assignment in the left column where the blue agent at node $b$ has now utility $1$ (as it only has green neighbors) and incentive to jump to node $e$ to increase its utility to $2$. 
This leads to the third assignment in the left column where the green agent at node $c$ has utility $1$ (as it only has red neighbors) and incentive to jump to node $f$ to increase its utility to $2$. 
This leads to the third assignment in the right column where the red agent at node $d$ has utility $1$ again, and would now prefer to jump back to node $a$ where it can get utility $2$. 
This leads to the second assignment in the right column where the blue agent at node $e$ has utility $1$ and would also now prefer to go back to node $b$ to get utility $2$. 
This leads to the first assignment in the right column where the green agent at node $f$ has been left with only red neighbors and now has incentive to jump to node $c$ again to get utility $2$, thus completing the cycle.
\end{proof}

\begin{figure}[t]
\centering
    \includegraphics[scale=0.8]{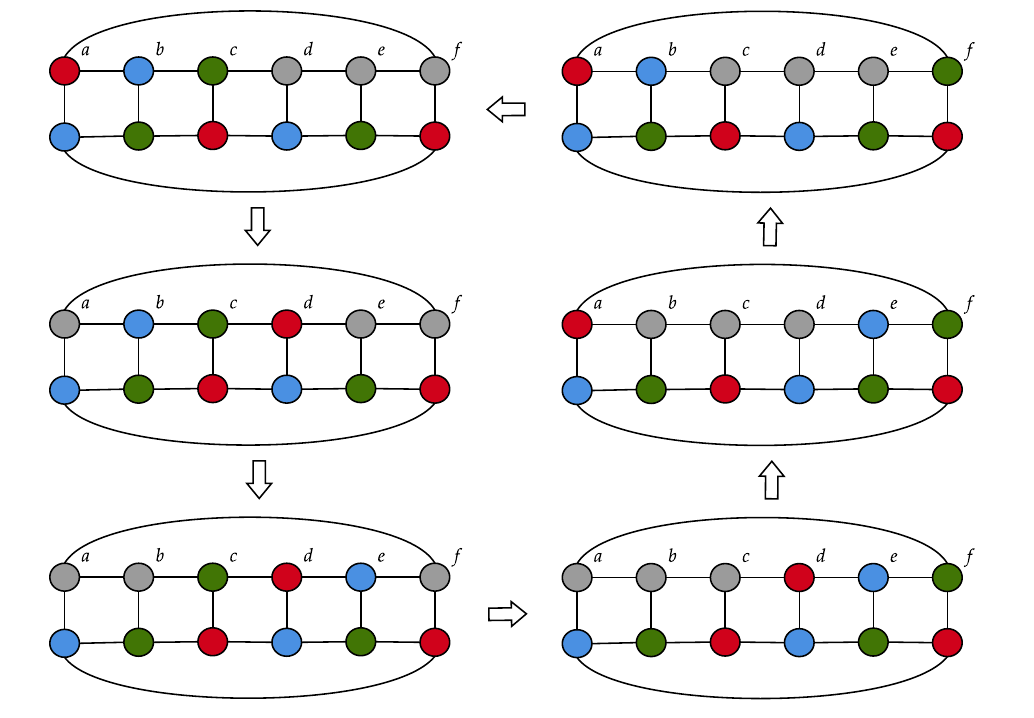}
    \caption{A game with an improving response cycle.}
    \label{fig:IRC}
\end{figure}

Next, we identify several cases where the game is indeed potential, and thus there is always at least one equilibrium assignment. In particular, we show that this is true when 
(1) there are only two types of agents, 
(2) there is a single empty node, 
(3) the graph is of degree at most $2$, and 
(4) there are two empty nodes and the graph is $3$-regular.  

\begin{theorem} \label{thm:existence:2agents}
The game is potential when there are only two types of agents.
\end{theorem}

\begin{proof}
We argue that the social welfare (the total utility of the agents) is a potential function. Observe that, if some agent $i$ jumps to an empty node, then it does so because its utility increases from $0$ to $1$. Since $i$ has only neighbors of its own type (so that its utility is $0$) in the initial assignment, the deviation of $i$ does not alter the utility of those agents. In addition, the utility of each of $i$'s new neighbors in the new assignment does not decrease in comparison to how much it was before; in particular, it either remains the same or increases from $0$ to $1$. Hence, the social welfare overall increases by at least $1$ after $i$ jumps, and is thus a potential function, leading to an equilibrium when it is maximized. 
\end{proof}

\begin{theorem} \label{thm:existence:1empty}
The game is potential when the graph contains just one empty node. 
\end{theorem}

\begin{proof}
We will prove by induction that, for every $\ell \in [k]$, there is no IRC containing an assignment $\bv$ in which the empty node, say $u$, has agents of $\ell$ different types surrounding it, that is, $\tau_u(\bv)=\ell$. Consequently, no IRC can exist in the Nash dynamics, thus showing that the game is potential.

\medskip
\noindent 
{\bf Base case: $\ell=1$.}
Suppose for the purpose of contradiction that there exists an IRC containing an assignment $\bv$, in which $|T_u(\bv)|=1$ for the empty node $u$. Without loss of generality, suppose all the agents adjacent to $u$ are of type $R$.  Consider the next assignment $\bv'$ in the IRC, and suppose agent $i$ is located at node $u$ in $\bv'$. 
Let $v$ be the node that $i$ occupies in $\bv$; that is, agent $i$ jumps from node $v$ to node $u$, implying that $u_i(\bv') > u_i(\bv)$. Note that $v$ is the empty node in the assignment $\bv'$. 

Agent $i$ cannot be of type $R$ since it gets utility $1$ after jumping to $u$, and it must have had utility $0$ by occupying $v$ in assignment $\bv$. If $i$ had no neighbors in the assignment $\bv$, then the only neighbor of $v$ in the graph is $u$, the empty node, but then $u$ has neighbors of two types, a contradiction. Therefore, agent $i$ must have had only neighbors of its own type when occupying $v$ in $\bv$. So, $i$ was segregated in the assignment $\bv$ and $|T_v(\bv')|=1$. 
Applying the same argument again, it follows that every assignment in the IRC has an empty node in which all neighboring agents are of the same type, and every jump is made by a segregated agent. Notice however, that each such jump reduces the number of segregated agents, leading to a contradiction.   

\medskip
\noindent
{\bf Induction hypothesis:}
Assume there is no IRC containing an assignment where the empty node has neighbors of $p$ different types, for all $p \in [\ell]$ for some $\ell < k$. 

\medskip
\noindent
{\bf Induction step:}
We will show that the hypothesis remains true for $\ell+1$.
Suppose instead that there is an IRC containing an assignment $\bv$ where the empty node $u$ has neighbors of exactly $\ell+1$ types, that is, $\tau_u(\bv) = \ell+1$. Consider the next assignment $\bv'$ in the IRC that is the result of agent $i$ jumping from node $v$ to $u$. Then, the empty node in $\bv'$ is $v$ and, by the inductive hypothesis, $\tau_v(\bv') > \ell$. Consider the following three cases: 
\begin{itemize}
\item $t(i) \in T_u(\bv)$.
Then $u_i(\bv') = \ell$ and $u_i(\bv) \geq \ell$, which contradicts the fact that $i$ is motivated to jump.

\item $t(i) \not\in T_u(\bv)$ and $t(i) \not\in T_v(\bv)$. 
Then $u_i(\bv') = \ell+1$ and $u_i(\bv) > \ell$, which contradicts the fact that $i$ is motivated to jump.

\item $t(i) \not\in T_u(\bv)$ and $t(i) \in T_v(\bv)$. 
Then $u_i(\bv') = \ell+1$, and $u_i(\bv) \geq \ell$.
For $i$ to be motivated to jump, it must be that $u_i(\bv) = \ell$. Since $t(i) \in T_v(\bv)$, this implies that $\tau_v(\bv') = \ell+1$. In other words, the empty node $v$ in assignment $\bv'$ has $\ell+1$ types in its neighborhood. 
\end{itemize}
From the above we see that the {\em number of monochromatic edges} in $\bv'$ is less than that in $\bv$. Applying the same argument repeatedly, it follows that for every assignment in the IRC, (a) the empty node has $\ell+1$ types in its neighborhood, and (b) the agent that deviates has at least one neighbor of its own type in its neighborhood, but no neighbors of its type at the empty node it jumps to, causing the number of monochromatic edges to decrease in successive assignments, which contradictions the existence of the IRC. 
\end{proof}

\begin{theorem} \label{thm:existence:degree2}
The game is potential when the graph is of degree at most $2$. 
\end{theorem}

\begin{proof}
For any assignment $\bv$, recall that $\CE(\bv)$ is the number of colorful edges that connect agents of different types. In addition, let $c(\bv)$ be the number of empty nodes that are adjacent to at most one type of agents, that is
$c(\bv) = |\{v \in V: v \text{ empty in } \bv \text{ and } \tau_v(\bv) \leq 1|\}$.
We will show that the function $\Phi(\bv)=2 \cdot \CE(\bv) + c(\bv)$ is a potential one. 

We will focus on two assignments $\bv$ and $\bv'$ that differ on the node occupied by a single red agent $i$. In particular, agent $i$ occupies a source node $s$ in $\bv$ and jumps to a destination node $d$ in $\bv'$. Note that $d$ is empty in $\bv$ and $s$ is empty in $\bv'$. We will show that $\Delta \Phi=\Phi(\bv') - \Phi(\bv) > 0$ when $u_i(\bv') > u_i(\bv)$. For simplicity, further define $\Delta \CE = \CE(\bv')- \CE(\bv)$ and $\Delta c = c(\bv')-c(\bv)$. We consider the following two cases depending on the utility that agent $i$ derives in assignment $\bv$.

\medskip
\noindent 
{\bf Case 1: $u_i(\bv)=0$.}
Since $i$ has only neighbors of its own type or no neighbors at all, the source node $s$ it occupies does not participate in any colorful edges. Since $i$ jumps from $s$ to $d$, $d$ must be adjacent to at least one non-red agent for $i$ to increase its utility. Hence, the number of colorful edges increases by at least $1$ when $i$ jumps, and thus $\Delta \CE \geq 1$. 
In $\bv'$, node $s$ is empty and is surrounded by at most one type of agents. Hence, $\Delta c$ may increase by $1$ due to $s$. If $d$ is surrounded by only one type of agents in $\bv$, then $\Delta c$ may decrease by $1$. 
Since the graph is of degree at most $2$ and $d$ is adjacent to a non-red agent, $d$ may be adjacent to at most one empty node, say $v$, and $\Delta c$ might further decrease by $1$ due to $v$ (when $i$ becomes adjacent to $v$ by jumping to $d$). Overall, $\Delta c \geq -1$, and thus $\Delta \Phi = 2\Delta \CE + \Delta c \geq 1$. 

\medskip
\noindent
{\bf Case 2: $u_i(\bv)=1$.}
In this case, the node $d$ where $i$ jumps to from $s$ must be adjacent to two agents of different types, say blue and green. This further implies that $d$ is not adjacent to $s$. 
\begin{itemize}
\item 
If $i$ has only one non-red neighbor in $\bv$, then there will be one more colorful edge in $\bv'$, and thus $\Delta \CE = 1$. Now observe that $\Delta c$ is not affected due to $d$ since it is adjacent to a blue and a green agent in $\bv$ where it is empty. Due to $s$, $\Delta c$ might increase by $1$ in case $s$ is adjacent to a non-red agent only, or not at all in case $s$ is adjacent to a red and a non-red agent. So, overall, $\Delta c \geq 0$, and $\Delta \Phi = 2\Delta \CE + \Delta c \ge 2$. 

\item
Otherwise, $i$ has two non-red neighbors of the same type in $\bv$. Hence, there is no change in the number of colorful edges when $i$ jumps from $s$ to $d$, that is, $\Delta \CE = 0$. Similarly to before, there is no change in $\Delta c$ due to $d$, but there is an increase of $1$ in $\Delta c$ due to $s$, and thus $\Delta c = 1$. Overall, $\Delta \Phi = 2\Delta \CE + \Delta c = 1$.
\end{itemize}
This completes the proof.
\end{proof}

Before proceeding with our next result, we first prove a lemma that will be useful. 

\begin{lemma} \label{lem:variety:potential:3-degree-1}
Let $\bv$ be a non-equilibrium assignment and suppose an agent $i$ jumps from node $s$ to node $d$ leading to assignment $\bv'$. Then,
\begin{enumerate}
    \item[(1)] $\tau_d(\bv) \geq \tau_s(\bv^\prime)$, and 
    \item[(2)] $\tau_d(\bv)=\tau_s(\bv^\prime)$  only if the jumping agent $i$'s utility increases exactly by 1, the number of monochromatic edges decreases, and no new monochromatic edges are created. 
\end{enumerate}
\end{lemma}

\begin{proof}
Without loss of generality, let the jumping agent $i$ be red. 
Suppose there are  agents of $x$ non-red types that are neighbors to $s$ in $\bv$.  Then there must be agents of at least $x+1$  non-red types that are neighbors to $d$ in $\bv$, as otherwise $i$ is not motivated to jump. So, we have  
$\tau_d(\bv) \geq x+1$ and 
$\tau_s(\bv^\prime) \le x+1$,  since $s$ could have red neighbors in $\bv'$. Therefore, $\tau_d(\bv) \ge \tau_s(\bv^\prime)$.

To show (2),  suppose $\tau_d(\bv) = \tau_s(\bv^\prime)$. Then they both equal $x+1$. This shows that $u_i(\bv)=x$ and $u_i(\bv') = x+1$, that is, the utility of $i$ increases exactly by 1, as claimed. 
Since $s$ had $x$ non-red types as neighbors in $\bv$, $\tau_s(\bv') = x+1$ implies that $s$ has at least one red neighbor in $\bv'$. Similarly, $\tau_d(\bv) = x+1$ implies that $d$ has no red neighbor in $\bv$. This implies that $s$ cannot be a neighbor of $d$. So, $s$ must have some  neighbor $p$ other than $d$ with a red agent. When agent $i$ jumps, the monochromatic edge between $s$ and $p$ is destroyed, and further, the jump does not create new monochromatic edges incident to $d$, since $d$ had no red neighbors. This proves (2). 
\end{proof}

\begin{theorem} \label{thm:existence:2empty3regular}
The game is potential when the graph has two empty nodes and is $3$-regular. 
\end{theorem}

\begin{proof}
For any assignment $\bv$, let $\Tau(\bv) = \tau_{v_1}(\bv) + \tau_{v_2}(\bv)$ be the total type-count of the two empty nodes $v_1$ and $v_2$. Also, let $M(\bv)$ be the number of monochromatic edges in $\bv$, and 
$$B(\bv) =
\begin{cases}
1, & \text{if the empty nodes are adjacent} 
\\
0, & \text{otherwise}. \\
\end{cases}$$
We will show that the function 
$$\Phi(\bv)=2(\Tau(\bv)+M(\bv))+B(\bv)$$
is a potential one. 

Let $\bv$ be an arbitrary non-equilibrium assignment, and suppose red agent $i$ jumps from a node $s$ to node $d$  which leads to the next assignment $\bv'$. For any function $f$ of an assignment, let $\Delta f = f(\bv') - f(\bv)$. We will show that
$$\Delta \Phi = \Phi(\bv') - \Phi(\bv) < 0$$

Observe that $d$ is an empty node in $\bv$ and $s$ is an empty node in $\bv'$. Since there are two empty nodes, there is another node, say $o$ that is empty in both $\bv$ and $\bv'$. We consider the following 4 cases for the location of $o$ in the neighborhood of $s$ and $d$. 

\medskip
\noindent
{\bf Case 1: $o$ is neither a neighbor of $s$ nor of $d$.}
In this case, there is no edge between the two empty nodes in either $\bv$ or $\bv'$, that is, $\Delta B = 0$. Also, $o$'s  type-count remains unchanged, that is, $\tau_o(\bv) = \tau_o(\bv')$. So, 
$\Delta \Tau = \tau_s(\bv^\prime) - \tau_d(\bv)$. Therefore, 
$$\Delta \Phi = 2(\Delta \Tau + \Delta M) + \Delta B = 2(\tau_s(\bv^\prime) - \tau_d(\bv) + \Delta M)$$
From Lemma~\ref{lem:variety:potential:3-degree-1}\hspace*{0.01in}(1)  we have $\tau_d(\bv) \ge \tau_s(\bv^\prime)$. Also, from Lemma~\ref{lem:variety:potential:3-degree-1}\hspace*{0.01in}(2), if $\tau_d(\bv) = \tau_s(\bv^\prime)$, then $\Delta M \le -1$ which implies that  $\Delta \Phi \le 2(0 - 1) = -2$. Therefore we assume that  $\tau_d(\bv) > \tau_s(\bv^\prime)$ and will show that $\Delta \Phi \le -1$, by showing that $\tau_s(\bv^\prime) - \tau_d(\bv) + \Delta M \le -1$, for all possible jumps. 

\begin{itemize}
\item
 $u_i(\bv)=0$ and $u_i(\bv') \ge 1$: If $u_i(\bv)=0$, then all the agents surrounding $s$ are of type red. This implies that there are at least 2 red agents adjacent to $s$ (as $d$, which is empty before the jump, could also be a neighbor of  $s$). So, at least 2 monochromatic red edges are destroyed because of the jump. After the jump, the utility of $i$ increases by at least 1, which implies there is at least 1 non-red agent that is a neighbor of $d$. So, at most 2 new monochromatic red edges could be created after the jump. Therefore, $\Delta M \le 0$ and $\tau_s(\bv^\prime) - \tau_d(\bv) + \Delta M \le -1$ is satisfied.

\item
 $u_i(\bv)=1$ and $u_i(\bv') \ge 2$: In this case, there is at least 1 non-red agent that is a neighbor of $s$ and at most 1 red agent that is a neighbor of $d$. $M(\bv')$ is always less than or equal to $ M(\bv)$ except when there are no monochromatic red edges incident on $s$ in $\bv$ and exactly 1 monochromatic red edge incident on $d$ in $\bv'$. Therefore, $\Delta M=1$ implies there are 2 agents of different non-red types and one red agent surrounding $d$, and $s$ is not a neighbor of  $d$. So $\tau_s(\bv^\prime) - \tau_d(\bv) + \Delta M =1-3+1= -1$. When $\Delta M \le 0$, $\tau_s(\bv^\prime) - \tau_d(\bv) + \Delta M \le -1+0= -1$.

\item
 $u_i(\bv)=2$ and $u_i(\bv') = 3$: No new monochromatic red edges will be created since all the neighbors of $d$ must be non-red to give agent $i$  utility 3. So $\Delta M \le 0$ and $\tau_s(\bv^\prime) - \tau_d(\bv) + \Delta M \le -1$.
 \end{itemize}

\noindent
{\bf Case 2: $o$ is a neighbor of $s$ but not $d$.}
Since  $s$ is an empty node in $\bv'$, we have $B(\bv')=1$. As  $d$ is an empty node in $\bv$, we have  $B(\bv)=0$. Hence $\Delta B = 1$. Since $o$ gains no new neighbor and loses a red neighbor in $\bv'$ 
compared to $\bv$, we have  $\tau_o(\bv') \le \tau_o(\bv)$ and $\Delta \tau_o \le 0$. So, $\Delta \Tau = (\tau_s(\bv^\prime) - \tau_d(\bv)) + \Delta \tau_o \le \tau_s(\bv^\prime) - \tau_d(\bv)$. Therefore, 
$$\Delta \Phi \le 2(\tau_s(\bv^\prime) - \tau_d(\bv) + \Delta M) + 1$$
From Lemma~\ref{lem:variety:potential:3-degree-1}\hspace*{0.01in}(1)  we have $\tau_d(\bv) \ge \tau_s(\bv^\prime)$. Also, from Lemma~\ref{lem:variety:potential:3-degree-1}\hspace*{0.01in}(2), if $\tau_d(\bv) = \tau_s(\bv^\prime)$, then $\Delta M \le -1$ which implies that  $\Delta \Phi \le 2(0 - 1)+1 = -1$. Therefore we assume that  $\tau_d(\bv) > \tau_s(\bv^\prime)$, that is, $\tau_s(\bv^\prime) - \tau_d(\bv) \le -1$. We will show that  $\tau_s(\bv^\prime) - \tau_d(\bv) + \Delta M \le -1$, for all possible jumps of the agent $i$. 

\begin{itemize}
\item 
$u_i(\bv)=0$ and $u_i(\bv') \ge 1$: There is at least $1$ non-red agent that is adjacent to $d$. This implies that there can be at most $2$ red agents that are neighbors of $d$. Also, $u_i(\bv)=0$ implies that all the agents surrounding $s$ are red agents. When $d$ is not a neighbor of $s$, there must be $2$ red agents connected to $s$ in $\bv$. So, $2$ monochromatic red edges are destroyed when $i$ jumps, and and at most $2$ monochromatic red edges are created. Hence, $\Delta M \le 0$. When $d$ is a neighbor to $s$, exactly $1$ red agent must be adjacent to $s$ in  $\bv$. After the jump $d$ is connected to the empty $s$ and at least $1$ non-red neighbor. So $1$ monochromatic red edge is destroyed and at most $1$ monochromatic red edge is created after the jump. Hence, $\Delta M \le 0$. Therefore, $\tau_s(\bv^\prime) - \tau_d(\bv) + \Delta M \le -1$.

\item
 $u_i(\bv) \ge 1$: Using the same analysis as in Case 1 we can show that $\tau_s(\bv^\prime) - \tau_d(\bv) + \Delta M \le -1$.
 \end{itemize}

\medskip
\noindent 
{\bf Case 3: $o$ is a neighbor of $d$ but not $s$.}
Since $s$ is an empty node in $\bv'$  and $d$ is an empty node in $\bv$, we have  $\Delta B = -1$. Therefore, $\Delta \Phi = 2(\Delta \Tau + \Delta M) - 1$.
To show that $\Delta \Phi \le -1$, it suffices to show that $\Delta \Tau + \Delta M \le 0$. 

For $i$ to be motivated to jump to $d$, since $d$ is adjacent to the empty node $o$, it must be that $d$ has at least 1 non-red neighbor and therefore at most one red neighbor in $\bv$. If $d$ has no red neighbors in $\bv$, then clearly no new monochromatic red edges will be created when $i$ jumps, so $\Delta M \le 0$. Suppose instead that $d$ has one non-red neighbor and one red neighbor in $\bv$. If $s$ is adjacent to $d$, then that red neighbor is in fact agent $i$, and when $i$ jumps to $d$, then no monochromatic red edges would be created, and $\Delta M \le 0$. If $s$ is not adjacent to $d$, it must be that $u_i(\bv)=0$, since $i$ gets utility 1 after jumping to $d$, and so $s$ must have 3 monochromatic edges incident on it that will be destroyed when $i$ jumps from $s$ to $d$, and only one new monochromatic edge incident on $d$ will be created. 
Therefore, $\Delta M \le 0$.

After the jump, a new red agent becomes a neighbor to $o$, as $o$ is connected to $d$. So, $0 \le \Delta \tau_o \le 1$. From Lemma~\ref{lem:variety:potential:3-degree-1}\hspace*{0.01in}(1) we know that $\tau_d(\bv) \ge \tau_s(\bv^\prime)$. If $\tau_s(\bv^\prime) - \tau_d(\bv) \le -1$, then $\Delta \Tau = (\tau_s(\bv^\prime) - \tau_d(\bv)) + \Delta \tau_o \le -1 + 1 = 0$. So, $\Delta \Tau + \Delta M \le 0$. Also, from Lemma~\ref{lem:variety:potential:3-degree-1}\hspace*{0.01in}(2), if $\tau_s(\bv^\prime) - \tau_d(\bv) = 0$, then $\Delta M \le -1$, and 
$$\Delta \Tau = (\tau_s(\bv^\prime) - \tau_d(\bv)) + \Delta \tau_o \le 0 + 1 = 1.$$ 
Therefore, 
$$\Delta \Tau + \Delta M \le 1 - 1 = 0.$$

\medskip
\noindent 
{\bf Case 4: $o$ is a neighbor of  both $s$ and $d$.}
Here, $\tau_o(\bv')=\tau_o(\bv)$, so $\Delta \tau_o=0$. There is an edge between the 2 empty nodes, $d$ and $o$, in $\bv$ and there is an edge between the 2 empty nodes, $s$ and $o$, in $\bv'$. So, $\Delta B  = 0$. 
Therefore, 
$$\Delta \Phi = 2(\Delta \Tau + \Delta M) + \Delta B = 2(\tau_s(\bv^\prime) - \tau_d(\bv) + \Delta M).$$

From Lemma~\ref{lem:variety:potential:3-degree-1}\hspace*{0.01in}(1)  we have $\tau_d(\bv) \ge \tau_s(\bv^\prime)$. Also, from Lemma~\ref{lem:variety:potential:3-degree-1}\hspace*{0.01in}(2), if $\tau_d(\bv) = \tau_s(\bv^\prime)$, then $\Delta M \le -1$ which implies that  $\Delta \Phi = 2(\tau_s(\bv^\prime) - \tau_d(\bv) + \Delta M) \le 2(0 + (-1)) = -2$. Therefore we assume that $\tau_d(\bv) > \tau_s(\bv^\prime)$ that is, $ \tau_s(\bv^\prime) - \tau_d(\bv) \le -1$ and will show that $\Delta \Phi \le -1$, by showing that $\Delta M \le 0$, for all possible jumps of agent $i$.

\begin{itemize}
\item
$u_i(\bv)=0$ and $u_i(\bv') \ge 1$: If $u_i(\bv)=0$, then either $s$ has two red neighbors and is not adjacent to $d$, or $s$ has one red neighbor and is adjacent to $d$, which is an empty node in $\bv$. In both cases $s$ has at least one red neighbor. As $d$ must have at least 1 non-red neighbor for agent $i$ to be motivated to jump, $d$ can have at most 1 red neighbor. When agent $i$ jumps, at least 1 monochromatic red edge is destroyed and at most 1 monochromatic red edge is created. So $\Delta M \le 0$.

\item 
$u_i(\bv)=1$ and $u_i(\bv') = 2$: In this case, $d$ must have 2 non-red neighbors of two different types, and since $d$ is adjacent to the empty node $o$, no new monochromatic red edges are created, when agent $i$ jumps. This implies that $\Delta M \le 0$. Notice that $u_i(\bv') = 3$ is not possible because the empty node $o$ is a neighbor of $d$.
\end{itemize}
The proof is now complete.
\end{proof}

Next we show that that even when there is an arbitrary number of empty nodes, an equilibrium assignment is guaranteed to exist for some families of graphs, in particular, for trees, cylinder graphs, and tori. For these cases, we explicitly construct an equilibrium rather than showing that the game is potential. In fact, recall that the graph in the game used to show that there is an IRC in the proof of Theorem~\ref{thm:IRC} is a cylinder, and thus constructing a potential function for such graphs is not possible.

The assignment we give will have one of the following two properties:

\begin{description}
    \item[\pOne ]\phantomsection\label{itm:prop1} $u_i(\bv) \geq 1$ for every agent $i$, and $\tau_u(\bv) \leq 1$ for any empty node $u$.
    \item[\po ]\phantomsection\label{itm:prop0} All empty nodes are adjacent only to red agents or other empty nodes, and $u_i(\bv)\geq 1$ for every non-red agent $i$.
\end{description} 

Clearly, an assignment $\bv$ satisfying \hyperref[itm:prop1]{\pOne} is an equilibrium assignment as no agent is motivated to jump. Suppose $\bv$ satisfies \hyperref[itm:prop0]{\po}. Notice that there may be red agents with utility $0$, but they cannot increase their utility by jumping to an empty node. Any non-red agent can only get utility $1$ by jumping to an empty node, since it already has utility $1$, does not have incentive to to jump.

\begin{theorem} \label{thm:equilibrium:trees}
There exists an equilibrium assignment when the graph is a tree.
\end{theorem}

\begin{proof}
We will use a very similar approach to the one used by \citet{Narayanan2023diversity} to derive an equilibrium assignment for a different utility function. First, fix a root node $r$ of degree $1$, and repeatedly remove leaf nodes until we have a tree with exactly $n+1$ nodes. Call this tree $G'=(V',E')$. Now we use the algorithm of \citet{Narayanan2023diversity} for a tree with one empty node, to find an assignment in $G'$, such that (a) the root node $r$ is empty, (b) every non-red agent $i$ has at least one neighbor of a different type, and hence $u_i(\bv) \geq 1$, and (c) a red agent occupies the node $v$ that is the unique neighbor to the empty root $r$. Since we use exactly the same algorithm, we do not repeat it here. We only observe that since the only empty node is $r$ and it is adjacent to a single red agent, $\tau_r(\bv) = 1$ and property \hyperref[itm:prop0]{\po} is satisfied by the assignment. Note that there may be some red agents with utility 0 in this assignment. See Figure~\ref{fig:equilibrium:trees}(a) for an example. 

Next we show how to modify the assignment for $G'$ to an assignment for $G$ that satisfies one of the properties \hyperref[itm:prop0]{\po} and \hyperref[itm:prop1]{\pOne}.  We start with the same placement of agents $\bv$ as in $G'$. Agents in this assignment on $G'$ may have acquired some new neighboring nodes, but since they are all empty nodes, the utility of agents remains the same in the graph $G$. However, while every empty node in $G$ is adjacent to at most one agent, some of them may now be adjacent to non-red agents, violating property \hyperref[itm:prop0]{\po}. See Figure~\ref{fig:equilibrium:trees}(b) for an example. There may be some red agents that have utility $0$ that are motivated to jump to such an empty node.  Let $X$ be the set of red agents with utility $0$ and $V_e$ be the set of empty nodes that are adjacent to non-red agents. 
\begin{itemize}
    \item If $|X| \leq |V_e|$, then let all the agents in $X$ jump to any $|X|$ nodes of $V_e$. Call the new assignment $\bv'$. Consider a red agent $i$ in the set $X$. Let $p$ be a neighbor of $i$ in $\bv$ and let $q$ be a neighbor in $\bv'$. First note that since $u_i(\bv) = 0$, agent $p$ must be red, and since any node in $V'$ is adjacent to a single non-red agent, agent $q$ must be non-red. Therefore, the jump does not decrease the utility of $p$ and may increase the utility of $q$. See Figure~\ref{fig:equilibrium:trees}(c)  for an illustration.  We also have $u_i(\bv') = 1$ for all red agents.  Finally, every empty node $u$ is adjacent to at most one agent, so that $\tau_u(\bv) \leq 1$. Therefore property \hyperref[itm:prop1]{\pOne} is satisfied. 
    
    \item If $|X| > |V_e|$, let an arbitrary set of  $|V_e|$ red agents from the set  $R$ jump to the empty nodes $V'$. As in the previous case, the jumps do not decrease the utility of any agent. Further, all empty nodes are adjacent to red agents. Therefore, property  \hyperref[itm:prop0]{\po} is satisfied. 
\end{itemize}
Since $\bv'$ satisfies either \hyperref[itm:prop0]{\po} or \hyperref[itm:prop1]{\pOne}, it is an equilibrium assignment. 
\end{proof}

\begin{figure}[t]
\centering
\begin{subfigure}{0.3\columnwidth}
\centering
\tikzset{every picture/.style={line width=0.75pt}} 

\begin{tikzpicture}[x=0.75pt,y=0.75pt,yscale=-1,xscale=1]

\draw    (190.77,102.16) -- (176.51,92.42) ;
\draw  [fill={rgb, 255:red, 208; green, 2; blue, 27 }  ,fill opacity=1 ] (165.14,88.52) .. controls (165.14,85.09) and (167.89,82.31) .. (171.29,82.31) .. controls (174.68,82.31) and (177.43,85.09) .. (177.43,88.52) .. controls (177.43,91.95) and (174.68,94.73) .. (171.29,94.73) .. controls (167.89,94.73) and (165.14,91.95) .. (165.14,88.52) -- cycle ;
\draw    (149.94,100.77) -- (166.14,92.42) ;
\draw    (161.61,122.34) -- (153.18,111.9) ;
\draw  [fill={rgb, 255:red, 208; green, 2; blue, 27 }  ,fill opacity=1 ] (143.11,106.61) .. controls (143.11,103.18) and (145.86,100.4) .. (149.25,100.4) .. controls (152.64,100.4) and (155.39,103.18) .. (155.39,106.61) .. controls (155.39,110.04) and (152.64,112.82) .. (149.25,112.82) .. controls (145.86,112.82) and (143.11,110.04) .. (143.11,106.61) -- cycle ;
\draw    (134.39,121.64) -- (145.41,110.51) ;
\draw  [fill={rgb, 255:red, 208; green, 2; blue, 27 }  ,fill opacity=1 ] (158.7,127.16) .. controls (158.7,123.73) and (161.45,120.94) .. (164.85,120.94) .. controls (168.24,120.94) and (170.99,123.73) .. (170.99,127.16) .. controls (170.99,130.59) and (168.24,133.37) .. (164.85,133.37) .. controls (161.45,133.37) and (158.7,130.59) .. (158.7,127.16) -- cycle ;
\draw    (205.68,123.73) -- (197.25,113.29) ;
\draw  [fill={rgb, 255:red, 208; green, 2; blue, 27 }  ,fill opacity=1 ] (187.18,108) .. controls (187.18,104.57) and (189.93,101.79) .. (193.32,101.79) .. controls (196.71,101.79) and (199.46,104.57) .. (199.46,108) .. controls (199.46,111.43) and (196.71,114.21) .. (193.32,114.21) .. controls (189.93,114.21) and (187.18,111.43) .. (187.18,108) -- cycle ;
\draw    (143.46,142.51) -- (135.04,132.08) ;
\draw  [fill={rgb, 255:red, 208; green, 2; blue, 27 }  ,fill opacity=1 ] (124.96,126.79) .. controls (124.96,123.36) and (127.71,120.57) .. (131.11,120.57) .. controls (134.5,120.57) and (137.25,123.36) .. (137.25,126.79) .. controls (137.25,130.22) and (134.5,133) .. (131.11,133) .. controls (127.71,133) and (124.96,130.22) .. (124.96,126.79) -- cycle ;
\draw    (220.58,144.6) -- (212.16,134.16) ;
\draw  [fill={rgb, 255:red, 208; green, 2; blue, 27 }  ,fill opacity=1 ] (202.08,128.87) .. controls (202.08,125.44) and (204.83,122.66) .. (208.22,122.66) .. controls (211.62,122.66) and (214.37,125.44) .. (214.37,128.87) .. controls (214.37,132.3) and (211.62,135.09) .. (208.22,135.09) .. controls (204.83,135.09) and (202.08,132.3) .. (202.08,128.87) -- cycle ;
\draw    (173.92,142.51) -- (170.03,130.68) ;
\draw  [fill={rgb, 255:red, 248; green, 231; blue, 28 }  ,fill opacity=1 ] (137.32,148.73) .. controls (137.32,145.29) and (140.07,142.51) .. (143.46,142.51) .. controls (146.85,142.51) and (149.6,145.29) .. (149.6,148.73) .. controls (149.6,152.16) and (146.85,154.94) .. (143.46,154.94) .. controls (140.07,154.94) and (137.32,152.16) .. (137.32,148.73) -- cycle ;
\draw  [fill={rgb, 255:red, 248; green, 231; blue, 28 }  ,fill opacity=1 ] (167.78,148.73) .. controls (167.78,145.29) and (170.53,142.51) .. (173.92,142.51) .. controls (177.31,142.51) and (180.06,145.29) .. (180.06,148.73) .. controls (180.06,152.16) and (177.31,154.94) .. (173.92,154.94) .. controls (170.53,154.94) and (167.78,152.16) .. (167.78,148.73) -- cycle ;
\draw  [fill={rgb, 255:red, 248; green, 231; blue, 28 }  ,fill opacity=1 ] (218.33,149.42) .. controls (218.33,145.99) and (221.08,143.21) .. (224.47,143.21) .. controls (227.86,143.21) and (230.61,145.99) .. (230.61,149.42) .. controls (230.61,152.85) and (227.86,155.63) .. (224.47,155.63) .. controls (221.08,155.63) and (218.33,152.85) .. (218.33,149.42) -- cycle ;
\draw  [fill={rgb, 255:red, 65; green, 117; blue, 5 }  ,fill opacity=1 ] (120.43,169.92) .. controls (120.43,166.49) and (123.18,163.71) .. (126.57,163.71) .. controls (129.96,163.71) and (132.71,166.49) .. (132.71,169.92) .. controls (132.71,173.35) and (129.96,176.13) .. (126.57,176.13) .. controls (123.18,176.13) and (120.43,173.35) .. (120.43,169.92) -- cycle ;
\draw    (127.26,164.08) -- (132.25,159.03) -- (138.28,152.95) ;
\draw  [fill={rgb, 255:red, 65; green, 117; blue, 5 }  ,fill opacity=1 ] (150.88,169.92) .. controls (150.88,166.49) and (153.64,163.71) .. (157.03,163.71) .. controls (160.42,163.71) and (163.17,166.49) .. (163.17,169.92) .. controls (163.17,173.35) and (160.42,176.13) .. (157.03,176.13) .. controls (153.64,176.13) and (150.88,173.35) .. (150.88,169.92) -- cycle ;
\draw    (157.72,164.08) -- (162.71,159.03) -- (168.74,152.95) ;
\draw  [fill={rgb, 255:red, 74; green, 144; blue, 226 }  ,fill opacity=1 ] (203.42,173.08) .. controls (203.42,169.64) and (206.17,166.86) .. (209.56,166.86) .. controls (212.96,166.86) and (215.71,169.64) .. (215.71,173.08) .. controls (215.71,176.51) and (212.96,179.29) .. (209.56,179.29) .. controls (206.17,179.29) and (203.42,176.51) .. (203.42,173.08) -- cycle ;
\draw    (209.56,166.86) -- (214.56,161.82) -- (221.23,155.04) ;
\draw  [fill={rgb, 255:red, 155; green, 155; blue, 155 }  ,fill opacity=1 ] (165.14,64.51) .. controls (165.14,61.08) and (167.89,58.3) .. (171.29,58.3) .. controls (174.68,58.3) and (177.43,61.08) .. (177.43,64.51) .. controls (177.43,67.95) and (174.68,70.73) .. (171.29,70.73) .. controls (167.89,70.73) and (165.14,67.95) .. (165.14,64.51) -- cycle ;
\draw    (171.29,70.73) -- (171.29,83.51) ;
\draw    (184,171) -- (178,153) ;
\draw  [fill={rgb, 255:red, 74; green, 144; blue, 226 }  ,fill opacity=1 ] (177.86,171) .. controls (177.86,167.57) and (180.61,164.79) .. (184,164.79) .. controls (187.39,164.79) and (190.14,167.57) .. (190.14,171) .. controls (190.14,174.43) and (187.39,177.21) .. (184,177.21) .. controls (180.61,177.21) and (177.86,174.43) .. (177.86,171) -- cycle ;
\draw    (237,174) -- (227.47,155.63) ;
\draw  [fill={rgb, 255:red, 208; green, 2; blue, 27 }  ,fill opacity=1 ] (230.86,174) .. controls (230.86,170.57) and (233.61,167.79) .. (237,167.79) .. controls (240.39,167.79) and (243.14,170.57) .. (243.14,174) .. controls (243.14,177.43) and (240.39,180.21) .. (237,180.21) .. controls (233.61,180.21) and (230.86,177.43) .. (230.86,174) -- cycle ;

\draw (177.66,49.09) node [anchor=north west][inner sep=0.75pt]  [font=\scriptsize]  {$\ r$};
\draw (179.8,77.9) node [anchor=north west][inner sep=0.75pt]  [font=\scriptsize]  {$\ v$};

\end{tikzpicture}

\caption{}
\label{fig:trees-a}
\end{subfigure}
\begin{subfigure}{0.3\columnwidth}
\centering
\tikzset{every picture/.style={line width=0.75pt}} 

\begin{tikzpicture}[x=0.75pt,y=0.75pt,yscale=-1,xscale=1]

\draw    (332.93,103.36) -- (318.67,93.62) ;
\draw  [fill={rgb, 255:red, 208; green, 2; blue, 27 }  ,fill opacity=1 ] (307.3,89.72) .. controls (307.3,86.29) and (310.05,83.51) .. (313.45,83.51) .. controls (316.84,83.51) and (319.59,86.29) .. (319.59,89.72) .. controls (319.59,93.15) and (316.84,95.93) .. (313.45,95.93) .. controls (310.05,95.93) and (307.3,93.15) .. (307.3,89.72) -- cycle ;
\draw    (292.1,101.97) -- (308.31,93.62) ;
\draw    (303.77,123.54) -- (295.34,113.1) ;
\draw  [fill={rgb, 255:red, 208; green, 2; blue, 27 }  ,fill opacity=1 ] (285.27,107.81) .. controls (285.27,104.38) and (288.02,101.6) .. (291.41,101.6) .. controls (294.81,101.6) and (297.56,104.38) .. (297.56,107.81) .. controls (297.56,111.24) and (294.81,114.02) .. (291.41,114.02) .. controls (288.02,114.02) and (285.27,111.24) .. (285.27,107.81) -- cycle ;
\draw    (276.55,122.84) -- (287.57,111.71) ;
\draw  [fill={rgb, 255:red, 155; green, 155; blue, 155 }  ,fill opacity=1 ] (252.26,149.23) .. controls (252.26,145.8) and (255.01,143.02) .. (258.4,143.02) .. controls (261.8,143.02) and (264.55,145.8) .. (264.55,149.23) .. controls (264.55,152.66) and (261.8,155.44) .. (258.4,155.44) .. controls (255.01,155.44) and (252.26,152.66) .. (252.26,149.23) -- cycle ;
\draw  [fill={rgb, 255:red, 208; green, 2; blue, 27 }  ,fill opacity=1 ] (300.87,128.36) .. controls (300.87,124.93) and (303.62,122.15) .. (307.01,122.15) .. controls (310.4,122.15) and (313.15,124.93) .. (313.15,128.36) .. controls (313.15,131.79) and (310.4,134.57) .. (307.01,134.57) .. controls (303.62,134.57) and (300.87,131.79) .. (300.87,128.36) -- cycle ;
\draw    (347.84,124.93) -- (339.41,114.49) ;
\draw  [fill={rgb, 255:red, 208; green, 2; blue, 27 }  ,fill opacity=1 ] (329.34,109.2) .. controls (329.34,105.77) and (332.09,102.99) .. (335.48,102.99) .. controls (338.87,102.99) and (341.62,105.77) .. (341.62,109.2) .. controls (341.62,112.63) and (338.87,115.42) .. (335.48,115.42) .. controls (332.09,115.42) and (329.34,112.63) .. (329.34,109.2) -- cycle ;
\draw    (285.62,143.71) -- (277.2,133.28) ;
\draw  [fill={rgb, 255:red, 208; green, 2; blue, 27 }  ,fill opacity=1 ] (267.12,127.99) .. controls (267.12,124.56) and (269.87,121.77) .. (273.27,121.77) .. controls (276.66,121.77) and (279.41,124.56) .. (279.41,127.99) .. controls (279.41,131.42) and (276.66,134.2) .. (273.27,134.2) .. controls (269.87,134.2) and (267.12,131.42) .. (267.12,127.99) -- cycle ;
\draw    (258.4,143.02) -- (269.42,131.89) ;
\draw    (362.74,145.8) -- (354.32,135.36) ;
\draw  [fill={rgb, 255:red, 208; green, 2; blue, 27 }  ,fill opacity=1 ] (344.24,130.07) .. controls (344.24,126.64) and (346.99,123.86) .. (350.39,123.86) .. controls (353.78,123.86) and (356.53,126.64) .. (356.53,130.07) .. controls (356.53,133.51) and (353.78,136.29) .. (350.39,136.29) .. controls (346.99,136.29) and (344.24,133.51) .. (344.24,130.07) -- cycle ;
\draw    (335.52,145.1) -- (346.54,133.97) ;
\draw    (316.08,143.71) -- (312.19,131.89) ;
\draw  [fill={rgb, 255:red, 248; green, 231; blue, 28 }  ,fill opacity=1 ] (279.48,149.93) .. controls (279.48,146.49) and (282.23,143.71) .. (285.62,143.71) .. controls (289.02,143.71) and (291.77,146.49) .. (291.77,149.93) .. controls (291.77,153.36) and (289.02,156.14) .. (285.62,156.14) .. controls (282.23,156.14) and (279.48,153.36) .. (279.48,149.93) -- cycle ;
\draw  [fill={rgb, 255:red, 248; green, 231; blue, 28 }  ,fill opacity=1 ] (309.94,149.93) .. controls (309.94,146.49) and (312.69,143.71) .. (316.08,143.71) .. controls (319.47,143.71) and (322.23,146.49) .. (322.23,149.93) .. controls (322.23,153.36) and (319.47,156.14) .. (316.08,156.14) .. controls (312.69,156.14) and (309.94,153.36) .. (309.94,149.93) -- cycle ;
\draw  [fill={rgb, 255:red, 155; green, 155; blue, 155 }  ,fill opacity=1 ] (329.38,151.32) .. controls (329.38,147.89) and (332.13,145.1) .. (335.52,145.1) .. controls (338.92,145.1) and (341.67,147.89) .. (341.67,151.32) .. controls (341.67,154.75) and (338.92,157.53) .. (335.52,157.53) .. controls (332.13,157.53) and (329.38,154.75) .. (329.38,151.32) -- cycle ;
\draw  [fill={rgb, 255:red, 248; green, 231; blue, 28 }  ,fill opacity=1 ] (360.49,150.62) .. controls (360.49,147.19) and (363.24,144.41) .. (366.63,144.41) .. controls (370.02,144.41) and (372.77,147.19) .. (372.77,150.62) .. controls (372.77,154.05) and (370.02,156.83) .. (366.63,156.83) .. controls (363.24,156.83) and (360.49,154.05) .. (360.49,150.62) -- cycle ;
\draw  [fill={rgb, 255:red, 155; green, 155; blue, 155 }  ,fill opacity=1 ] (247.72,192.36) .. controls (247.72,188.93) and (250.47,186.15) .. (253.87,186.15) .. controls (257.26,186.15) and (260.01,188.93) .. (260.01,192.36) .. controls (260.01,195.8) and (257.26,198.58) .. (253.87,198.58) .. controls (250.47,198.58) and (247.72,195.8) .. (247.72,192.36) -- cycle ;
\draw    (281.09,186.85) -- (272.66,176.41) ;
\draw  [fill={rgb, 255:red, 65; green, 117; blue, 5 }  ,fill opacity=1 ] (262.59,171.12) .. controls (262.59,167.69) and (265.34,164.91) .. (268.73,164.91) .. controls (272.12,164.91) and (274.87,167.69) .. (274.87,171.12) .. controls (274.87,174.55) and (272.12,177.33) .. (268.73,177.33) .. controls (265.34,177.33) and (262.59,174.55) .. (262.59,171.12) -- cycle ;
\draw    (253.87,186.15) -- (264.88,175.02) ;
\draw  [fill={rgb, 255:red, 155; green, 155; blue, 155 }  ,fill opacity=1 ] (274.94,193.06) .. controls (274.94,189.63) and (277.69,186.85) .. (281.09,186.85) .. controls (284.48,186.85) and (287.23,189.63) .. (287.23,193.06) .. controls (287.23,196.49) and (284.48,199.27) .. (281.09,199.27) .. controls (277.69,199.27) and (274.94,196.49) .. (274.94,193.06) -- cycle ;
\draw    (269.42,165.28) -- (274.41,160.24) -- (280.44,154.15) ;
\draw    (311.55,186.85) -- (303.12,176.41) ;
\draw  [fill={rgb, 255:red, 65; green, 117; blue, 5 }  ,fill opacity=1 ] (293.05,171.12) .. controls (293.05,167.69) and (295.8,164.91) .. (299.19,164.91) .. controls (302.58,164.91) and (305.33,167.69) .. (305.33,171.12) .. controls (305.33,174.55) and (302.58,177.33) .. (299.19,177.33) .. controls (295.8,177.33) and (293.05,174.55) .. (293.05,171.12) -- cycle ;
\draw  [fill={rgb, 255:red, 155; green, 155; blue, 155 }  ,fill opacity=1 ] (305.4,193.06) .. controls (305.4,189.63) and (308.15,186.85) .. (311.55,186.85) .. controls (314.94,186.85) and (317.69,189.63) .. (317.69,193.06) .. controls (317.69,196.49) and (314.94,199.27) .. (311.55,199.27) .. controls (308.15,199.27) and (305.4,196.49) .. (305.4,193.06) -- cycle ;
\draw    (299.88,165.28) -- (304.87,160.24) -- (310.9,154.15) ;
\draw    (363.39,189.63) -- (354.97,179.19) ;
\draw  [fill={rgb, 255:red, 74; green, 144; blue, 226 }  ,fill opacity=1 ] (345.58,174.28) .. controls (345.58,170.84) and (348.33,168.06) .. (351.73,168.06) .. controls (355.12,168.06) and (357.87,170.84) .. (357.87,174.28) .. controls (357.87,177.71) and (355.12,180.49) .. (351.73,180.49) .. controls (348.33,180.49) and (345.58,177.71) .. (345.58,174.28) -- cycle ;
\draw  [fill={rgb, 255:red, 155; green, 155; blue, 155 }  ,fill opacity=1 ] (357.25,195.84) .. controls (357.25,192.41) and (360,189.63) .. (363.39,189.63) .. controls (366.78,189.63) and (369.53,192.41) .. (369.53,195.84) .. controls (369.53,199.27) and (366.78,202.06) .. (363.39,202.06) .. controls (360,202.06) and (357.25,199.27) .. (357.25,195.84) -- cycle ;
\draw    (351.73,168.06) -- (356.72,163.02) -- (363.39,156.24) ;
\draw  [fill={rgb, 255:red, 155; green, 155; blue, 155 }  ,fill opacity=1 ] (273.65,236.89) .. controls (273.65,233.46) and (276.4,230.68) .. (279.79,230.68) .. controls (283.18,230.68) and (285.93,233.46) .. (285.93,236.89) .. controls (285.93,240.32) and (283.18,243.1) .. (279.79,243.1) .. controls (276.4,243.1) and (273.65,240.32) .. (273.65,236.89) -- cycle ;
\draw    (307.01,231.37) -- (298.58,220.94) ;
\draw  [fill={rgb, 255:red, 155; green, 155; blue, 155 }  ,fill opacity=1 ] (288.51,215.65) .. controls (288.51,212.22) and (291.26,209.44) .. (294.65,209.44) .. controls (298.05,209.44) and (300.8,212.22) .. (300.8,215.65) .. controls (300.8,219.08) and (298.05,221.86) .. (294.65,221.86) .. controls (291.26,221.86) and (288.51,219.08) .. (288.51,215.65) -- cycle ;
\draw    (279.79,230.68) -- (290.81,219.55) ;
\draw  [fill={rgb, 255:red, 155; green, 155; blue, 155 }  ,fill opacity=1 ] (300.87,237.59) .. controls (300.87,234.16) and (303.62,231.37) .. (307.01,231.37) .. controls (310.4,231.37) and (313.15,234.16) .. (313.15,237.59) .. controls (313.15,241.02) and (310.4,243.8) .. (307.01,243.8) .. controls (303.62,243.8) and (300.87,241.02) .. (300.87,237.59) -- cycle ;
\draw    (295.99,209.11) -- (307.01,197.98) ;
\draw    (375.7,209.81) -- (367.28,199.37) ;
\draw  [fill={rgb, 255:red, 155; green, 155; blue, 155 }  ,fill opacity=1 ] (369.56,216.02) .. controls (369.56,212.59) and (372.31,209.81) .. (375.7,209.81) .. controls (379.1,209.81) and (381.85,212.59) .. (381.85,216.02) .. controls (381.85,219.45) and (379.1,222.23) .. (375.7,222.23) .. controls (372.31,222.23) and (369.56,219.45) .. (369.56,216.02) -- cycle ;
\draw  [fill={rgb, 255:red, 155; green, 155; blue, 155 }  ,fill opacity=1 ] (307.3,64.51) .. controls (307.3,61.08) and (310.05,58.3) .. (313.45,58.3) .. controls (316.84,58.3) and (319.59,61.08) .. (319.59,64.51) .. controls (319.59,67.95) and (316.84,70.73) .. (313.45,70.73) .. controls (310.05,70.73) and (307.3,67.95) .. (307.3,64.51) -- cycle ;
\draw    (313.45,70.73) -- (313.45,83.51) ;
\draw    (326,172) -- (320,154) ;
\draw  [fill={rgb, 255:red, 74; green, 144; blue, 226 }  ,fill opacity=1 ] (319.86,172) .. controls (319.86,168.57) and (322.61,165.79) .. (326,165.79) .. controls (329.39,165.79) and (332.14,168.57) .. (332.14,172) .. controls (332.14,175.43) and (329.39,178.21) .. (326,178.21) .. controls (322.61,178.21) and (319.86,175.43) .. (319.86,172) -- cycle ;
\draw    (378,174) -- (368.47,155.63) ;
\draw  [fill={rgb, 255:red, 208; green, 2; blue, 27 }  ,fill opacity=1 ] (371.86,174) .. controls (371.86,170.57) and (374.61,167.79) .. (378,167.79) .. controls (381.39,167.79) and (384.14,170.57) .. (384.14,174) .. controls (384.14,177.43) and (381.39,180.21) .. (378,180.21) .. controls (374.61,180.21) and (371.86,177.43) .. (371.86,174) -- cycle ;

\draw (341.51,99.86) node [anchor=north west][inner sep=0.75pt]  [font=\scriptsize]  {$\ v_{2}$};
\draw (342.73,197.04) node [anchor=north west][inner sep=0.75pt]  [font=\scriptsize]  {$\ v_{5}$};
\draw (323.03,53.89) node [anchor=north west][inner sep=0.75pt]  [font=\scriptsize]  {$\ r$};
\draw (323.03,80.3) node [anchor=north west][inner sep=0.75pt]  [font=\scriptsize]  {$\ v$};
\draw (313.55,196.46) node [anchor=north west][inner sep=0.75pt]  [font=\scriptsize]  {$\ v_{4}$};
\draw (267.13,95.32) node [anchor=north west][inner sep=0.75pt]  [font=\scriptsize]  {$\ v_{1}$};
\draw (262.01,197.76) node [anchor=north west][inner sep=0.75pt]  [font=\scriptsize]  {$\ v_{3}$};
\end{tikzpicture}

\caption{}
\label{fig:trees-b}
\end{subfigure}
\begin{subfigure}{0.3\columnwidth}
\centering
\tikzset{every picture/.style={line width=0.75pt}} 

\begin{tikzpicture}[x=0.75pt,y=0.75pt,yscale=-1,xscale=1]

\draw    (486.85,104.56) -- (472.59,94.82) ;
\draw  [fill={rgb, 255:red, 155; green, 155; blue, 155 }  ,fill opacity=1 ] (461.22,90.92) .. controls (461.22,87.49) and (463.97,84.71) .. (467.37,84.71) .. controls (470.76,84.71) and (473.51,87.49) .. (473.51,90.92) .. controls (473.51,94.35) and (470.76,97.14) .. (467.37,97.14) .. controls (463.97,97.14) and (461.22,94.35) .. (461.22,90.92) -- cycle ;
\draw    (446.02,103.17) -- (462.22,94.82) ;
\draw    (457.69,124.74) -- (449.26,114.3) ;
\draw  [fill={rgb, 255:red, 155; green, 155; blue, 155 }  ,fill opacity=1 ] (439.19,109.01) .. controls (439.19,105.58) and (441.94,102.8) .. (445.33,102.8) .. controls (448.72,102.8) and (451.47,105.58) .. (451.47,109.01) .. controls (451.47,112.44) and (448.72,115.22) .. (445.33,115.22) .. controls (441.94,115.22) and (439.19,112.44) .. (439.19,109.01) -- cycle ;
\draw    (430.47,124.04) -- (441.49,112.91) ;
\draw  [fill={rgb, 255:red, 155; green, 155; blue, 155 }  ,fill opacity=1 ] (406.18,150.43) .. controls (406.18,147) and (408.93,144.22) .. (412.32,144.22) .. controls (415.72,144.22) and (418.47,147) .. (418.47,150.43) .. controls (418.47,153.86) and (415.72,156.64) .. (412.32,156.64) .. controls (408.93,156.64) and (406.18,153.86) .. (406.18,150.43) -- cycle ;
\draw  [fill={rgb, 255:red, 208; green, 2; blue, 27 }  ,fill opacity=1 ] (454.79,129.56) .. controls (454.79,126.13) and (457.54,123.35) .. (460.93,123.35) .. controls (464.32,123.35) and (467.07,126.13) .. (467.07,129.56) .. controls (467.07,132.99) and (464.32,135.77) .. (460.93,135.77) .. controls (457.54,135.77) and (454.79,132.99) .. (454.79,129.56) -- cycle ;
\draw    (501.76,126.13) -- (493.33,115.69) ;
\draw  [fill={rgb, 255:red, 155; green, 155; blue, 155 }  ,fill opacity=1 ] (483.26,110.4) .. controls (483.26,106.97) and (486.01,104.19) .. (489.4,104.19) .. controls (492.79,104.19) and (495.54,106.97) .. (495.54,110.4) .. controls (495.54,113.83) and (492.79,116.62) .. (489.4,116.62) .. controls (486.01,116.62) and (483.26,113.83) .. (483.26,110.4) -- cycle ;
\draw    (439.54,144.91) -- (431.12,134.48) ;
\draw  [fill={rgb, 255:red, 208; green, 2; blue, 27 }  ,fill opacity=1 ] (421.04,129.19) .. controls (421.04,125.76) and (423.79,122.97) .. (427.19,122.97) .. controls (430.58,122.97) and (433.33,125.76) .. (433.33,129.19) .. controls (433.33,132.62) and (430.58,135.4) .. (427.19,135.4) .. controls (423.79,135.4) and (421.04,132.62) .. (421.04,129.19) -- cycle ;
\draw    (412.32,144.22) -- (423.34,133.09) ;
\draw    (516.66,147) -- (508.24,136.56) ;
\draw  [fill={rgb, 255:red, 208; green, 2; blue, 27 }  ,fill opacity=1 ] (498.16,131.27) .. controls (498.16,127.84) and (500.91,125.06) .. (504.31,125.06) .. controls (507.7,125.06) and (510.45,127.84) .. (510.45,131.27) .. controls (510.45,134.71) and (507.7,137.49) .. (504.31,137.49) .. controls (500.91,137.49) and (498.16,134.71) .. (498.16,131.27) -- cycle ;
\draw    (489.44,146.3) -- (500.46,135.17) ;
\draw    (470,144.91) -- (466.11,133.09) ;
\draw  [fill={rgb, 255:red, 248; green, 231; blue, 28 }  ,fill opacity=1 ] (433.4,151.13) .. controls (433.4,147.69) and (436.15,144.91) .. (439.54,144.91) .. controls (442.93,144.91) and (445.69,147.69) .. (445.69,151.13) .. controls (445.69,154.56) and (442.93,157.34) .. (439.54,157.34) .. controls (436.15,157.34) and (433.4,154.56) .. (433.4,151.13) -- cycle ;
\draw  [fill={rgb, 255:red, 248; green, 231; blue, 28 }  ,fill opacity=1 ] (463.86,151.13) .. controls (463.86,147.69) and (466.61,144.91) .. (470,144.91) .. controls (473.39,144.91) and (476.14,147.69) .. (476.14,151.13) .. controls (476.14,154.56) and (473.39,157.34) .. (470,157.34) .. controls (466.61,157.34) and (463.86,154.56) .. (463.86,151.13) -- cycle ;
\draw  [fill={rgb, 255:red, 155; green, 155; blue, 155 }  ,fill opacity=1 ] (483.3,152.52) .. controls (483.3,149.09) and (486.05,146.3) .. (489.44,146.3) .. controls (492.84,146.3) and (495.59,149.09) .. (495.59,152.52) .. controls (495.59,155.95) and (492.84,158.73) .. (489.44,158.73) .. controls (486.05,158.73) and (483.3,155.95) .. (483.3,152.52) -- cycle ;
\draw  [fill={rgb, 255:red, 248; green, 231; blue, 28 }  ,fill opacity=1 ] (514.41,151.82) .. controls (514.41,148.39) and (517.16,145.61) .. (520.55,145.61) .. controls (523.94,145.61) and (526.69,148.39) .. (526.69,151.82) .. controls (526.69,155.25) and (523.94,158.03) .. (520.55,158.03) .. controls (517.16,158.03) and (514.41,155.25) .. (514.41,151.82) -- cycle ;
\draw  [fill={rgb, 255:red, 155; green, 155; blue, 155 }  ,fill opacity=1 ] (401.64,193.57) .. controls (401.64,190.13) and (404.39,187.35) .. (407.79,187.35) .. controls (411.18,187.35) and (413.93,190.13) .. (413.93,193.57) .. controls (413.93,197) and (411.18,199.78) .. (407.79,199.78) .. controls (404.39,199.78) and (401.64,197) .. (401.64,193.57) -- cycle ;
\draw    (435.01,188.05) -- (426.58,177.61) ;
\draw  [fill={rgb, 255:red, 65; green, 117; blue, 5 }  ,fill opacity=1 ] (416.51,172.32) .. controls (416.51,168.89) and (419.26,166.11) .. (422.65,166.11) .. controls (426.04,166.11) and (428.79,168.89) .. (428.79,172.32) .. controls (428.79,175.75) and (426.04,178.53) .. (422.65,178.53) .. controls (419.26,178.53) and (416.51,175.75) .. (416.51,172.32) -- cycle ;
\draw    (407.79,187.35) -- (418.8,176.22) ;
\draw  [fill={rgb, 255:red, 208; green, 2; blue, 27 }  ,fill opacity=1 ] (428.86,194.26) .. controls (428.86,190.83) and (431.61,188.05) .. (435.01,188.05) .. controls (438.4,188.05) and (441.15,190.83) .. (441.15,194.26) .. controls (441.15,197.69) and (438.4,200.47) .. (435.01,200.47) .. controls (431.61,200.47) and (428.86,197.69) .. (428.86,194.26) -- cycle ;
\draw    (423.34,166.48) -- (428.33,161.44) -- (434.36,155.35) ;
\draw    (465.46,188.05) -- (457.04,177.61) ;
\draw  [fill={rgb, 255:red, 65; green, 117; blue, 5 }  ,fill opacity=1 ] (446.97,172.32) .. controls (446.97,168.89) and (449.72,166.11) .. (453.11,166.11) .. controls (456.5,166.11) and (459.25,168.89) .. (459.25,172.32) .. controls (459.25,175.75) and (456.5,178.53) .. (453.11,178.53) .. controls (449.72,178.53) and (446.97,175.75) .. (446.97,172.32) -- cycle ;
\draw  [fill={rgb, 255:red, 208; green, 2; blue, 27 }  ,fill opacity=1 ] (459.32,194.26) .. controls (459.32,190.83) and (462.07,188.05) .. (465.46,188.05) .. controls (468.86,188.05) and (471.61,190.83) .. (471.61,194.26) .. controls (471.61,197.69) and (468.86,200.47) .. (465.46,200.47) .. controls (462.07,200.47) and (459.32,197.69) .. (459.32,194.26) -- cycle ;
\draw    (453.8,166.48) -- (458.79,161.44) -- (464.82,155.35) ;
\draw    (517.31,190.83) -- (508.89,180.39) ;
\draw  [fill={rgb, 255:red, 74; green, 144; blue, 226 }  ,fill opacity=1 ] (499.5,175.48) .. controls (499.5,172.04) and (502.25,169.26) .. (505.64,169.26) .. controls (509.04,169.26) and (511.79,172.04) .. (511.79,175.48) .. controls (511.79,178.91) and (509.04,181.69) .. (505.64,181.69) .. controls (502.25,181.69) and (499.5,178.91) .. (499.5,175.48) -- cycle ;
\draw  [fill={rgb, 255:red, 208; green, 2; blue, 27 }  ,fill opacity=1 ] (511.17,197.04) .. controls (511.17,193.61) and (513.92,190.83) .. (517.31,190.83) .. controls (520.7,190.83) and (523.45,193.61) .. (523.45,197.04) .. controls (523.45,200.47) and (520.7,203.26) .. (517.31,203.26) .. controls (513.92,203.26) and (511.17,200.47) .. (511.17,197.04) -- cycle ;
\draw    (505.64,169.26) -- (510.64,164.22) -- (517.31,157.44) ;
\draw  [fill={rgb, 255:red, 155; green, 155; blue, 155 }  ,fill opacity=1 ] (427.57,238.09) .. controls (427.57,234.66) and (430.32,231.88) .. (433.71,231.88) .. controls (437.1,231.88) and (439.85,234.66) .. (439.85,238.09) .. controls (439.85,241.52) and (437.1,244.3) .. (433.71,244.3) .. controls (430.32,244.3) and (427.57,241.52) .. (427.57,238.09) -- cycle ;
\draw    (460.93,232.57) -- (452.5,222.14) ;
\draw  [fill={rgb, 255:red, 155; green, 155; blue, 155 }  ,fill opacity=1 ] (442.43,216.85) .. controls (442.43,213.42) and (445.18,210.64) .. (448.57,210.64) .. controls (451.96,210.64) and (454.72,213.42) .. (454.72,216.85) .. controls (454.72,220.28) and (451.96,223.06) .. (448.57,223.06) .. controls (445.18,223.06) and (442.43,220.28) .. (442.43,216.85) -- cycle ;
\draw    (433.71,231.88) -- (444.73,220.75) ;
\draw  [fill={rgb, 255:red, 155; green, 155; blue, 155 }  ,fill opacity=1 ] (454.79,238.79) .. controls (454.79,235.36) and (457.54,232.57) .. (460.93,232.57) .. controls (464.32,232.57) and (467.07,235.36) .. (467.07,238.79) .. controls (467.07,242.22) and (464.32,245) .. (460.93,245) .. controls (457.54,245) and (454.79,242.22) .. (454.79,238.79) -- cycle ;
\draw    (449.91,210.31) -- (460.93,199.18) ;
\draw    (529.62,211.01) -- (521.2,200.57) ;
\draw  [fill={rgb, 255:red, 155; green, 155; blue, 155 }  ,fill opacity=1 ] (523.48,217.22) .. controls (523.48,213.79) and (526.23,211.01) .. (529.62,211.01) .. controls (533.02,211.01) and (535.77,213.79) .. (535.77,217.22) .. controls (535.77,220.65) and (533.02,223.43) .. (529.62,223.43) .. controls (526.23,223.43) and (523.48,220.65) .. (523.48,217.22) -- cycle ;
\draw  [fill={rgb, 255:red, 155; green, 155; blue, 155 }  ,fill opacity=1 ] (461.22,65.71) .. controls (461.22,62.28) and (463.97,59.5) .. (467.37,59.5) .. controls (470.76,59.5) and (473.51,62.28) .. (473.51,65.71) .. controls (473.51,69.15) and (470.76,71.93) .. (467.37,71.93) .. controls (463.97,71.93) and (461.22,69.15) .. (461.22,65.71) -- cycle ;
\draw    (467.37,71.93) -- (467.37,84.71) ;
\draw    (479,173) -- (473,155) ;
\draw  [fill={rgb, 255:red, 74; green, 144; blue, 226 }  ,fill opacity=1 ] (472.86,173) .. controls (472.86,169.57) and (475.61,166.79) .. (479,166.79) .. controls (482.39,166.79) and (485.14,169.57) .. (485.14,173) .. controls (485.14,176.43) and (482.39,179.21) .. (479,179.21) .. controls (475.61,179.21) and (472.86,176.43) .. (472.86,173) -- cycle ;
\draw    (533,176) -- (523.47,157.63) ;
\draw  [fill={rgb, 255:red, 208; green, 2; blue, 27 }  ,fill opacity=1 ] (526.86,176) .. controls (526.86,172.57) and (529.61,169.79) .. (533,169.79) .. controls (536.39,169.79) and (539.14,172.57) .. (539.14,176) .. controls (539.14,179.43) and (536.39,182.21) .. (533,182.21) .. controls (529.61,182.21) and (526.86,179.43) .. (526.86,176) -- cycle ;

\draw (495.43,101.06) node [anchor=north west][inner sep=0.75pt]  [font=\scriptsize]  {$\ v_{2}$};
\draw (493.65,197.25) node [anchor=north west][inner sep=0.75pt]  [font=\scriptsize]  {$\ v_{5}$};
\draw (475.88,53.89) node [anchor=north west][inner sep=0.75pt]  [font=\scriptsize]  {$\ r$};
\draw (478.02,79.1) node [anchor=north west][inner sep=0.75pt]  [font=\scriptsize]  {$\ v$};
\draw (416.99,102.58) node [anchor=north west][inner sep=0.75pt]  [font=\scriptsize]  {$\ v_{1}$};
\draw (467.46,197.66) node [anchor=north west][inner sep=0.75pt]  [font=\scriptsize]  {$\ v_{4}$};
\draw (416.15,198.66) node [anchor=north west][inner sep=0.75pt]  [font=\scriptsize]  {$\ v_{3}$};
\end{tikzpicture}
\caption{}
\label{fig:trees-c}
\end{subfigure}
\caption{(a) The assignment $\bv$ for the tree $G'$ with a single empty node. Note that all non-red agents have utility at least 1, but some red agents have utility 0. (b) The same assignment for the tree $G$. Note that empty nodes $v_3,v_4$ and $v_5$ are adjacent to non-red agents (c) The assignment $\bv'$:  red agents at $v,v_1$ and $v_2$ with utility 0 in $\bv$ jump to nodes $v_3,v_4$ and $v_5$ to obtain utility 1 in $\bv'$. }
\label{fig:equilibrium:trees}
\end{figure}
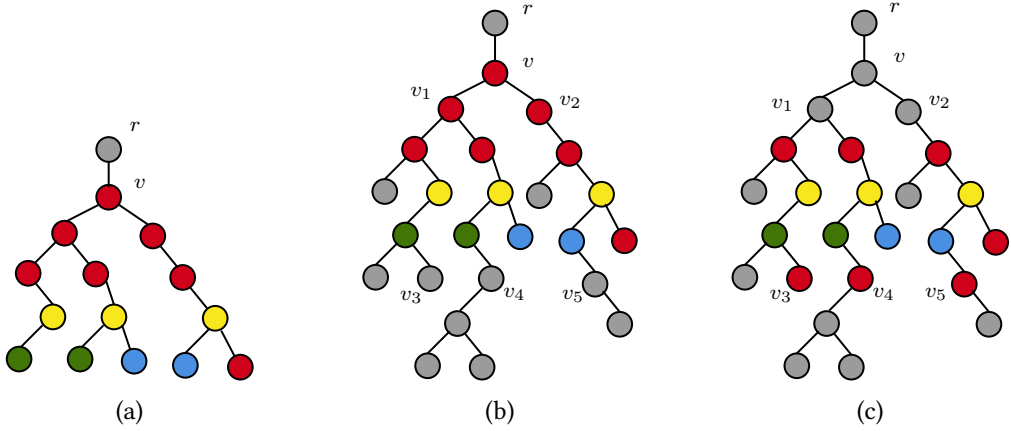

\begin{theorem} \label{thm:equilibrium:cylinders}
There exists an equilibrium assignment when the graph is a cylinder.
\end{theorem}

\begin{proof}
Let the dimensions of the cylinder be $2 \times m$, that is, there are $m$ columns and 2 rows for a total of $2m$ nodes. Since we know  from Theorems \ref{thm:existence:1empty},  \ref{thm:existence:2empty3regular}, and\ref{thm:existence:2agents} that the game is potential if there are at most 2 empty nodes or only 2 types of agents, we assume that there are at least three empty nodes, and the number of types  $k$ is at least  3. Recall that the $n$ agents are partitioned into $k$ types $T_1,T_2,\ldots,T_k$ such that $|T_j|=n_j$ for $j \in [k]$. We assume that $n_1 \leq n_2 \leq ...\leq n_k$.
Since it's simple to construct an equilibrium assignment if $n \leq 3$, we assume below that $n \ge 4$. We show explicitly how to construct an equilibrium assignment; there are three main cases, based on the value of $n_k$ the largest number of agents of a single type. 

\begin{figure}[t]
\centering

\tikzset{every picture/.style={line width=0.75pt}} 

\begin{tikzpicture}[x=0.75pt,y=0.75pt,yscale=-0.8,xscale=0.8]

\draw    (84.45,174.47) -- (109.31,174.47) ;
\draw  [fill={rgb, 255:red, 126; green, 211; blue, 33 }  ,fill opacity=1 ] (67.87,174.47) .. controls (67.87,170.04) and (71.58,166.45) .. (76.16,166.45) .. controls (80.74,166.45) and (84.45,170.04) .. (84.45,174.47) .. controls (84.45,178.9) and (80.74,182.49) .. (76.16,182.49) .. controls (71.58,182.49) and (67.87,178.9) .. (67.87,174.47) -- cycle ;
\draw  [fill={rgb, 255:red, 65; green, 117; blue, 5 }  ,fill opacity=1 ] (109.31,174.47) .. controls (109.31,170.04) and (113.02,166.45) .. (117.6,166.45) .. controls (122.18,166.45) and (125.89,170.04) .. (125.89,174.47) .. controls (125.89,178.9) and (122.18,182.49) .. (117.6,182.49) .. controls (113.02,182.49) and (109.31,178.9) .. (109.31,174.47) -- cycle ;
\draw  [fill={rgb, 255:red, 189; green, 16; blue, 224 }  ,fill opacity=1 ] (150.75,174.47) .. controls (150.75,170.04) and (154.46,166.45) .. (159.04,166.45) .. controls (163.62,166.45) and (167.33,170.04) .. (167.33,174.47) .. controls (167.33,178.9) and (163.62,182.49) .. (159.04,182.49) .. controls (154.46,182.49) and (150.75,178.9) .. (150.75,174.47) -- cycle ;
\draw    (125.89,174.47) -- (150.75,174.47) ;
\draw    (270.85,174.47) -- (295.71,174.47) ;
\draw  [fill={rgb, 255:red, 74; green, 144; blue, 226 }  ,fill opacity=1 ] (254.27,174.47) .. controls (254.27,170.04) and (257.98,166.45) .. (262.56,166.45) .. controls (267.13,166.45) and (270.85,170.04) .. (270.85,174.47) .. controls (270.85,178.9) and (267.13,182.49) .. (262.56,182.49) .. controls (257.98,182.49) and (254.27,178.9) .. (254.27,174.47) -- cycle ;
\draw  [fill={rgb, 255:red, 208; green, 2; blue, 27 }  ,fill opacity=1 ] (295.71,174.47) .. controls (295.71,170.04) and (299.42,166.45) .. (304,166.45) .. controls (308.57,166.45) and (312.28,170.04) .. (312.28,174.47) .. controls (312.28,178.9) and (308.57,182.49) .. (304,182.49) .. controls (299.42,182.49) and (295.71,178.9) .. (295.71,174.47) -- cycle ;
\draw  [fill={rgb, 255:red, 245; green, 166; blue, 35 }  ,fill opacity=1 ] (337.15,174.47) .. controls (337.15,170.04) and (340.86,166.45) .. (345.44,166.45) .. controls (350.01,166.45) and (353.72,170.04) .. (353.72,174.47) .. controls (353.72,178.9) and (350.01,182.49) .. (345.44,182.49) .. controls (340.86,182.49) and (337.15,178.9) .. (337.15,174.47) -- cycle ;
\draw    (312.28,174.47) -- (337.15,174.47) ;
\draw  [fill={rgb, 255:red, 139; green, 87; blue, 42 }  ,fill opacity=1 ] (67.87,212.47) .. controls (67.87,208.04) and (71.58,204.45) .. (76.16,204.45) .. controls (80.74,204.45) and (84.45,208.04) .. (84.45,212.47) .. controls (84.45,216.9) and (80.74,220.49) .. (76.16,220.49) .. controls (71.58,220.49) and (67.87,216.9) .. (67.87,212.47) -- cycle ;
\draw  [fill={rgb, 255:red, 155; green, 155; blue, 155 }  ,fill opacity=1 ] (150.75,211.57) .. controls (150.75,207.14) and (154.46,203.56) .. (159.04,203.56) .. controls (163.62,203.56) and (167.33,207.14) .. (167.33,211.57) .. controls (167.33,216) and (163.62,219.59) .. (159.04,219.59) .. controls (154.46,219.59) and (150.75,216) .. (150.75,211.57) -- cycle ;
\draw    (84.45,174.47) -- (109.31,174.47) ;
\draw    (159.04,182.49) -- (159.04,203.56) ;
\draw  [fill={rgb, 255:red, 184; green, 233; blue, 134 }  ,fill opacity=1 ] (295.71,211.57) .. controls (295.71,207.14) and (299.42,203.56) .. (304,203.56) .. controls (308.57,203.56) and (312.28,207.14) .. (312.28,211.57) .. controls (312.28,216) and (308.57,219.59) .. (304,219.59) .. controls (299.42,219.59) and (295.71,216) .. (295.71,211.57) -- cycle ;
\draw    (76.16,204.45) -- (76.16,182.49) ;
\draw    (117.6,182.49) -- (117.6,203.87) ;
\draw    (262.56,182.49) -- (262.56,203.87) ;
\draw  [fill={rgb, 255:red, 126; green, 211; blue, 33 }  ,fill opacity=1 ] (109.31,211.88) .. controls (109.31,207.46) and (113.02,203.87) .. (117.6,203.87) .. controls (122.18,203.87) and (125.89,207.46) .. (125.89,211.88) .. controls (125.89,216.31) and (122.18,219.9) .. (117.6,219.9) .. controls (113.02,219.9) and (109.31,216.31) .. (109.31,211.88) -- cycle ;
\draw    (125.89,211.88) -- (150.75,211.57) ;
\draw    (345.44,182.49) -- (345.44,203.87) ;
\draw    (304,182.49) -- (304,203.56) ;
\draw  [fill={rgb, 255:red, 80; green, 227; blue, 194 }  ,fill opacity=1 ] (254.27,211.88) .. controls (254.27,207.46) and (257.98,203.87) .. (262.56,203.87) .. controls (267.13,203.87) and (270.85,207.46) .. (270.85,211.88) .. controls (270.85,216.31) and (267.13,219.9) .. (262.56,219.9) .. controls (257.98,219.9) and (254.27,216.31) .. (254.27,211.88) -- cycle ;
\draw  [fill={rgb, 255:red, 208; green, 2; blue, 27 }  ,fill opacity=1 ] (337.15,211.88) .. controls (337.15,207.46) and (340.86,203.87) .. (345.44,203.87) .. controls (350.01,203.87) and (353.72,207.46) .. (353.72,211.88) .. controls (353.72,216.31) and (350.01,219.9) .. (345.44,219.9) .. controls (340.86,219.9) and (337.15,216.31) .. (337.15,211.88) -- cycle ;
\draw    (270.85,211.88) -- (295.71,211.57) ;
\draw    (76.16,166.45) .. controls (96.44,131.19) and (408.8,131.19) .. (428.71,166.45) ;
\draw    (428.71,219.9) .. controls (403.55,252.38) and (98.19,251.49) .. (76.16,220.49) ;
\draw    (84.45,212.47) -- (109.31,211.88) ;
\draw    (312.28,211.57) -- (337.15,211.88) ;
\draw   (174.65,193.69) .. controls (174.65,193.09) and (175.12,192.61) .. (175.7,192.61) .. controls (176.28,192.61) and (176.75,193.09) .. (176.75,193.69) .. controls (176.75,194.28) and (176.28,194.76) .. (175.7,194.76) .. controls (175.12,194.76) and (174.65,194.28) .. (174.65,193.69) -- cycle ;
\draw   (184.27,193.69) .. controls (184.27,193.09) and (184.74,192.61) .. (185.32,192.61) .. controls (185.9,192.61) and (186.37,193.09) .. (186.37,193.69) .. controls (186.37,194.28) and (185.9,194.76) .. (185.32,194.76) .. controls (184.74,194.76) and (184.27,194.28) .. (184.27,193.69) -- cycle ;
\draw   (194.46,193.93) .. controls (194.46,193.33) and (194.93,192.85) .. (195.51,192.85) .. controls (196.09,192.85) and (196.56,193.33) .. (196.56,193.93) .. controls (196.56,194.52) and (196.09,195.01) .. (195.51,195.01) .. controls (194.93,195.01) and (194.46,194.52) .. (194.46,193.93) -- cycle ;
\draw  [fill={rgb, 255:red, 144; green, 19; blue, 254 }  ,fill opacity=1 ] (211.77,174.47) .. controls (211.77,170.04) and (215.48,166.45) .. (220.06,166.45) .. controls (224.64,166.45) and (228.35,170.04) .. (228.35,174.47) .. controls (228.35,178.9) and (224.64,182.49) .. (220.06,182.49) .. controls (215.48,182.49) and (211.77,178.9) .. (211.77,174.47) -- cycle ;
\draw    (228.35,174.47) -- (253.21,174.47) ;
\draw  [fill={rgb, 255:red, 155; green, 155; blue, 155 }  ,fill opacity=1 ] (211.77,211.57) .. controls (211.77,207.14) and (215.48,203.56) .. (220.06,203.56) .. controls (224.64,203.56) and (228.35,207.14) .. (228.35,211.57) .. controls (228.35,216) and (224.64,219.59) .. (220.06,219.59) .. controls (215.48,219.59) and (211.77,216) .. (211.77,211.57) -- cycle ;
\draw    (220.06,182.49) -- (220.06,203.56) ;
\draw    (228.35,211.57) -- (253.21,211.88) ;
\draw   (203.51,193.69) .. controls (203.51,193.09) and (203.98,192.61) .. (204.56,192.61) .. controls (205.13,192.61) and (205.6,193.09) .. (205.6,193.69) .. controls (205.6,194.28) and (205.13,194.76) .. (204.56,194.76) .. controls (203.98,194.76) and (203.51,194.28) .. (203.51,193.69) -- cycle ;
\draw  [fill={rgb, 255:red, 248; green, 231; blue, 28 }  ,fill opacity=1 ] (379.12,174.47) .. controls (379.12,170.04) and (382.83,166.45) .. (387.4,166.45) .. controls (391.98,166.45) and (395.69,170.04) .. (395.69,174.47) .. controls (395.69,178.9) and (391.98,182.49) .. (387.4,182.49) .. controls (382.83,182.49) and (379.12,178.9) .. (379.12,174.47) -- cycle ;
\draw    (354.25,174.47) -- (379.12,174.47) ;
\draw    (387.4,182.49) -- (387.4,203.87) ;
\draw  [fill={rgb, 255:red, 245; green, 166; blue, 35 }  ,fill opacity=1 ] (379.12,211.88) .. controls (379.12,207.46) and (382.83,203.87) .. (387.4,203.87) .. controls (391.98,203.87) and (395.69,207.46) .. (395.69,211.88) .. controls (395.69,216.31) and (391.98,219.9) .. (387.4,219.9) .. controls (382.83,219.9) and (379.12,216.31) .. (379.12,211.88) -- cycle ;
\draw    (354.25,211.57) -- (379.12,211.88) ;
\draw  [fill={rgb, 255:red, 139; green, 87; blue, 42 }  ,fill opacity=1 ] (420.42,174.47) .. controls (420.42,170.04) and (424.13,166.45) .. (428.71,166.45) .. controls (433.29,166.45) and (437,170.04) .. (437,174.47) .. controls (437,178.9) and (433.29,182.49) .. (428.71,182.49) .. controls (424.13,182.49) and (420.42,178.9) .. (420.42,174.47) -- cycle ;
\draw    (428.71,182.49) -- (428.71,203.87) ;
\draw  [fill={rgb, 255:red, 248; green, 231; blue, 28 }  ,fill opacity=1 ] (420.42,211.88) .. controls (420.42,207.46) and (424.13,203.87) .. (428.71,203.87) .. controls (433.29,203.87) and (437,207.46) .. (437,211.88) .. controls (437,216.31) and (433.29,219.9) .. (428.71,219.9) .. controls (424.13,219.9) and (420.42,216.31) .. (420.42,211.88) -- cycle ;
\draw    (395.69,174.47) -- (420.56,174.47) ;
\draw    (395.69,211.88) -- (420.56,211.88) ;

\draw (309.2,160.4) node [anchor=north west][inner sep=0.75pt]  [font=\scriptsize]  {$\ v_{1}$};
\draw (267.2,160.4) node [anchor=north west][inner sep=0.75pt]  [font=\scriptsize]  {$\ v_{2}$};
\draw (265.56,216.28) node [anchor=north west][inner sep=0.75pt]  [font=\scriptsize]  {$\ v_{3}$};
\draw (120.6,216.28) node [anchor=north west][inner sep=0.75pt]  [font=\scriptsize]  {$\ v_{4}$};

\end{tikzpicture}

\caption{An example of an equilibrium assignment in a cylinder graph, when $n_k=2$.}
\label{fig:equilibrium:cylindercase2}
\end{figure}
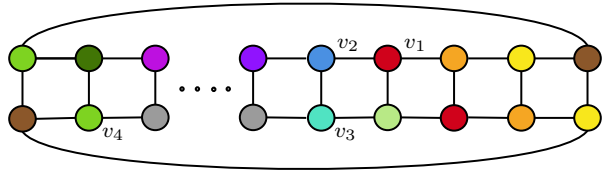

\medskip
\noindent
{\bf Case 1: $\mathbf{n_k=1} $.}
This implies $n_1 = n_2 = ...= n_k=1$, since $n_k \ge n_j$ for all $j < k$. 
If $n < m$, place all agents in the top row at consecutive locations. Notice that every agent has utility 2 except the agents at the ends of the row, which have utility 1. When $n < m-1$, any empty node offers a utility of at most 1, so no agent wants to jump. When $n = m-1$, the sole empty node in the top row offers utility 2, and other empty nodes have utility 1. However, only the agents adjacent to this empty node have utility 1, and they would not improve their utility by jumping to the empty node. Thus the assignment is an equilibrium. 
Otherwise, $n \geq m$ and we place $m$ agents in the top row and the remaining $n-m$ agents at consecutive locations in the bottom row. Every agent has utility at least 2 and no empty node offers utility $>2$, so we have an equilibrium assignment. 

\medskip
\noindent
{\bf Case 2: $\mathbf{n_k=2}$.}
Place the agents in a sequence such that the types of the agents in the sequence are as follows: 
$T_k, T_{k-1},\ldots,T_1, T_k,T_{k-1},\ldots,T_z$, where $z$ is the smallest index with $|T_z| = 2$. 

\medskip
\noindent 
{\bf Case 2~(a):} If $n \le m$, then all agents fit in one row of the cylinder. Place the sequence of agents in the top row of the cylinder. Notice that every agent has a utility of 2 except the agents at both the ends of the sequence. All the empty nodes in the bottom row can offer utility of at most 1 to any agent. If there is only one empty node in the top row, it offers utility at most 2, but only utility 1 to the agents adjacent to it. Therefore no agent wants to jump. If there is more than one empty node in the top row, then all such nodes offer utility at most 1 to any agent, and therefore, no agent is motivated to jump. In both cases, we have an equilibrium. 

\medskip
\noindent 
{\bf Case 2~(b):} Otherwise, we have $m<n$. Place the first $m$ agents of the sequence in the top row of the cylinder in anticlockwise direction starting from any node, say $v_1$. Let $v_2$ be the last node of the top row to which an agent is assigned and it is adjacent to $v_1$. Let $v_3$ be the node in the bottom row that is adjacent to $v_2$. Start assigning agents to the bottom row from $v_3$ and assign the remaining agents of the sequence in anticlockwise direction, with the last agent being assigned to node $v_4$ (see Figure~\ref{fig:equilibrium:cylindercase2}).  This leaves the consecutive nodes from $v_4$ going counterclockwise to $v_3$  empty. Since there are at least three empty nodes, every empty node is adjacent to at least one empty node, and offers utility at most 2. We will show that all agents have utility at least 2; in some cases we need to adjust the above assignment of agents. 

\begin{itemize}

\item If $k > m+1$, then every agent in the top row, as well as every agent in the bottom row except those at $v_3$ and $v_4$ has utility of at least 2 since it has agents of 2 different types adjacent to it in the same row. 
Also the agents at $v_3$ and $v_4$ have utility 2 as they have one neighbor of a different type in the bottom row, and one in the top row.

\item If $k=m+1$, then every agent starting from the one that is next to $v_3$, in the bottom row, will have a neighbor of the same type in the top row. So, move every agent in the bottom row by one node in clockwise direction. Now notice that every agent in the bottom row will also have a utility of at least 2. 

\item If $k=m$, then $v_2$ will be occupied by an agent of type $T_1$ and $v_3$ will be occupied by an agent of type $T_k$. Observe that the agent in $v_4$ has a utility of just 1. Move every agent in the bottom row by one node in the clockwise direction. Now notice that every agent in the bottom row will also have a utility of at least 2.

\item If $k < m$, every agent in the top row has utility at least $2$ because of neighbors with in the top row or neighbors in both the top and bottom rows, and every agent in the bottom row has utility at least 2. 
\end{itemize}

\begin{figure}[t]
\centering

\tikzset{every picture/.style={line width=0.75pt}} 

\begin{tikzpicture}[x=0.75pt,y=0.75pt,yscale=-0.8,xscale=0.8]

\draw    (84.45,174.47) -- (109.31,174.47) ;
\draw  [fill={rgb, 255:red, 208; green, 2; blue, 27 }  ,fill opacity=1 ] (67.87,174.47) .. controls (67.87,170.04) and (71.58,166.45) .. (76.16,166.45) .. controls (80.74,166.45) and (84.45,170.04) .. (84.45,174.47) .. controls (84.45,178.9) and (80.74,182.49) .. (76.16,182.49) .. controls (71.58,182.49) and (67.87,178.9) .. (67.87,174.47) -- cycle ;
\draw  [fill={rgb, 255:red, 74; green, 144; blue, 226 }  ,fill opacity=1 ] (109.31,174.47) .. controls (109.31,170.04) and (113.02,166.45) .. (117.6,166.45) .. controls (122.18,166.45) and (125.89,170.04) .. (125.89,174.47) .. controls (125.89,178.9) and (122.18,182.49) .. (117.6,182.49) .. controls (113.02,182.49) and (109.31,178.9) .. (109.31,174.47) -- cycle ;
\draw  [fill={rgb, 255:red, 248; green, 231; blue, 28 }  ,fill opacity=1 ] (150.75,174.47) .. controls (150.75,170.04) and (154.46,166.45) .. (159.04,166.45) .. controls (163.62,166.45) and (167.33,170.04) .. (167.33,174.47) .. controls (167.33,178.9) and (163.62,182.49) .. (159.04,182.49) .. controls (154.46,182.49) and (150.75,178.9) .. (150.75,174.47) -- cycle ;
\draw    (125.89,174.47) -- (150.75,174.47) ;
\draw    (270.85,174.47) -- (295.71,174.47) ;
\draw  [fill={rgb, 255:red, 65; green, 117; blue, 5 }  ,fill opacity=1 ] (254.27,174.47) .. controls (254.27,170.04) and (257.98,166.45) .. (262.56,166.45) .. controls (267.13,166.45) and (270.85,170.04) .. (270.85,174.47) .. controls (270.85,178.9) and (267.13,182.49) .. (262.56,182.49) .. controls (257.98,182.49) and (254.27,178.9) .. (254.27,174.47) -- cycle ;
\draw  [fill={rgb, 255:red, 208; green, 2; blue, 27 }  ,fill opacity=1 ] (295.71,174.47) .. controls (295.71,170.04) and (299.42,166.45) .. (304,166.45) .. controls (308.57,166.45) and (312.28,170.04) .. (312.28,174.47) .. controls (312.28,178.9) and (308.57,182.49) .. (304,182.49) .. controls (299.42,182.49) and (295.71,178.9) .. (295.71,174.47) -- cycle ;
\draw  [fill={rgb, 255:red, 155; green, 155; blue, 155 }  ,fill opacity=1 ] (337.15,174.47) .. controls (337.15,170.04) and (340.86,166.45) .. (345.44,166.45) .. controls (350.01,166.45) and (353.72,170.04) .. (353.72,174.47) .. controls (353.72,178.9) and (350.01,182.49) .. (345.44,182.49) .. controls (340.86,182.49) and (337.15,178.9) .. (337.15,174.47) -- cycle ;
\draw    (312.28,174.47) -- (337.15,174.47) ;
\draw  [fill={rgb, 255:red, 208; green, 2; blue, 27 }  ,fill opacity=1 ] (67.87,212.47) .. controls (67.87,208.04) and (71.58,204.45) .. (76.16,204.45) .. controls (80.74,204.45) and (84.45,208.04) .. (84.45,212.47) .. controls (84.45,216.9) and (80.74,220.49) .. (76.16,220.49) .. controls (71.58,220.49) and (67.87,216.9) .. (67.87,212.47) -- cycle ;
\draw  [fill={rgb, 255:red, 189; green, 16; blue, 224 }  ,fill opacity=1 ] (150.75,211.57) .. controls (150.75,207.14) and (154.46,203.56) .. (159.04,203.56) .. controls (163.62,203.56) and (167.33,207.14) .. (167.33,211.57) .. controls (167.33,216) and (163.62,219.59) .. (159.04,219.59) .. controls (154.46,219.59) and (150.75,216) .. (150.75,211.57) -- cycle ;
\draw    (84.45,174.47) -- (109.31,174.47) ;
\draw    (159.04,182.49) -- (159.04,203.56) ;
\draw  [fill={rgb, 255:red, 208; green, 2; blue, 27 }  ,fill opacity=1 ] (295.71,211.57) .. controls (295.71,207.14) and (299.42,203.56) .. (304,203.56) .. controls (308.57,203.56) and (312.28,207.14) .. (312.28,211.57) .. controls (312.28,216) and (308.57,219.59) .. (304,219.59) .. controls (299.42,219.59) and (295.71,216) .. (295.71,211.57) -- cycle ;
\draw    (76.16,204.45) -- (76.16,182.49) ;
\draw    (117.6,182.49) -- (117.6,203.87) ;
\draw    (262.56,182.49) -- (262.56,203.87) ;
\draw  [fill={rgb, 255:red, 248; green, 231; blue, 28 }  ,fill opacity=1 ] (109.31,211.88) .. controls (109.31,207.46) and (113.02,203.87) .. (117.6,203.87) .. controls (122.18,203.87) and (125.89,207.46) .. (125.89,211.88) .. controls (125.89,216.31) and (122.18,219.9) .. (117.6,219.9) .. controls (113.02,219.9) and (109.31,216.31) .. (109.31,211.88) -- cycle ;
\draw    (125.89,211.88) -- (150.75,211.57) ;
\draw    (345.44,182.49) -- (345.44,203.87) ;
\draw    (304,182.49) -- (304,203.56) ;
\draw  [fill={rgb, 255:red, 208; green, 2; blue, 27 }  ,fill opacity=1 ] (254.27,211.88) .. controls (254.27,207.46) and (257.98,203.87) .. (262.56,203.87) .. controls (267.13,203.87) and (270.85,207.46) .. (270.85,211.88) .. controls (270.85,216.31) and (267.13,219.9) .. (262.56,219.9) .. controls (257.98,219.9) and (254.27,216.31) .. (254.27,211.88) -- cycle ;
\draw  [fill={rgb, 255:red, 208; green, 2; blue, 27 }  ,fill opacity=1 ] (337.15,211.88) .. controls (337.15,207.46) and (340.86,203.87) .. (345.44,203.87) .. controls (350.01,203.87) and (353.72,207.46) .. (353.72,211.88) .. controls (353.72,216.31) and (350.01,219.9) .. (345.44,219.9) .. controls (340.86,219.9) and (337.15,216.31) .. (337.15,211.88) -- cycle ;
\draw    (270.85,211.88) -- (295.71,211.57) ;
\draw    (76.16,166.45) .. controls (96.44,131.19) and (408.8,131.19) .. (428.71,166.45) ;
\draw    (428.71,219.9) .. controls (403.55,252.38) and (98.19,251.49) .. (76.16,220.49) ;
\draw    (84.45,212.47) -- (109.31,211.88) ;
\draw    (312.28,211.57) -- (337.15,211.88) ;
\draw   (174.65,193.69) .. controls (174.65,193.09) and (175.12,192.61) .. (175.7,192.61) .. controls (176.28,192.61) and (176.75,193.09) .. (176.75,193.69) .. controls (176.75,194.28) and (176.28,194.76) .. (175.7,194.76) .. controls (175.12,194.76) and (174.65,194.28) .. (174.65,193.69) -- cycle ;
\draw   (184.27,193.69) .. controls (184.27,193.09) and (184.74,192.61) .. (185.32,192.61) .. controls (185.9,192.61) and (186.37,193.09) .. (186.37,193.69) .. controls (186.37,194.28) and (185.9,194.76) .. (185.32,194.76) .. controls (184.74,194.76) and (184.27,194.28) .. (184.27,193.69) -- cycle ;
\draw   (194.46,193.93) .. controls (194.46,193.33) and (194.93,192.85) .. (195.51,192.85) .. controls (196.09,192.85) and (196.56,193.33) .. (196.56,193.93) .. controls (196.56,194.52) and (196.09,195.01) .. (195.51,195.01) .. controls (194.93,195.01) and (194.46,194.52) .. (194.46,193.93) -- cycle ;
\draw  [fill={rgb, 255:red, 189; green, 16; blue, 224 }  ,fill opacity=1 ] (211.77,174.47) .. controls (211.77,170.04) and (215.48,166.45) .. (220.06,166.45) .. controls (224.64,166.45) and (228.35,170.04) .. (228.35,174.47) .. controls (228.35,178.9) and (224.64,182.49) .. (220.06,182.49) .. controls (215.48,182.49) and (211.77,178.9) .. (211.77,174.47) -- cycle ;
\draw    (228.35,174.47) -- (253.21,174.47) ;
\draw  [fill={rgb, 255:red, 65; green, 117; blue, 5 }  ,fill opacity=1 ] (211.77,211.57) .. controls (211.77,207.14) and (215.48,203.56) .. (220.06,203.56) .. controls (224.64,203.56) and (228.35,207.14) .. (228.35,211.57) .. controls (228.35,216) and (224.64,219.59) .. (220.06,219.59) .. controls (215.48,219.59) and (211.77,216) .. (211.77,211.57) -- cycle ;
\draw    (220.06,182.49) -- (220.06,203.56) ;
\draw    (228.35,211.57) -- (253.21,211.88) ;
\draw   (203.51,193.69) .. controls (203.51,193.09) and (203.98,192.61) .. (204.56,192.61) .. controls (205.13,192.61) and (205.6,193.09) .. (205.6,193.69) .. controls (205.6,194.28) and (205.13,194.76) .. (204.56,194.76) .. controls (203.98,194.76) and (203.51,194.28) .. (203.51,193.69) -- cycle ;
\draw  [fill={rgb, 255:red, 155; green, 155; blue, 155 }  ,fill opacity=1 ] (379.12,174.47) .. controls (379.12,170.04) and (382.83,166.45) .. (387.4,166.45) .. controls (391.98,166.45) and (395.69,170.04) .. (395.69,174.47) .. controls (395.69,178.9) and (391.98,182.49) .. (387.4,182.49) .. controls (382.83,182.49) and (379.12,178.9) .. (379.12,174.47) -- cycle ;
\draw    (354.25,174.47) -- (379.12,174.47) ;
\draw    (387.4,182.49) -- (387.4,203.87) ;
\draw  [fill={rgb, 255:red, 155; green, 155; blue, 155 }  ,fill opacity=1 ] (379.12,211.88) .. controls (379.12,207.46) and (382.83,203.87) .. (387.4,203.87) .. controls (391.98,203.87) and (395.69,207.46) .. (395.69,211.88) .. controls (395.69,216.31) and (391.98,219.9) .. (387.4,219.9) .. controls (382.83,219.9) and (379.12,216.31) .. (379.12,211.88) -- cycle ;
\draw    (354.25,211.57) -- (379.12,211.88) ;
\draw  [fill={rgb, 255:red, 155; green, 155; blue, 155 }  ,fill opacity=1 ] (420.42,174.47) .. controls (420.42,170.04) and (424.13,166.45) .. (428.71,166.45) .. controls (433.29,166.45) and (437,170.04) .. (437,174.47) .. controls (437,178.9) and (433.29,182.49) .. (428.71,182.49) .. controls (424.13,182.49) and (420.42,178.9) .. (420.42,174.47) -- cycle ;
\draw    (428.71,182.49) -- (428.71,203.87) ;
\draw  [fill={rgb, 255:red, 155; green, 155; blue, 155 }  ,fill opacity=1 ] (420.42,211.88) .. controls (420.42,207.46) and (424.13,203.87) .. (428.71,203.87) .. controls (433.29,203.87) and (437,207.46) .. (437,211.88) .. controls (437,216.31) and (433.29,219.9) .. (428.71,219.9) .. controls (424.13,219.9) and (420.42,216.31) .. (420.42,211.88) -- cycle ;
\draw    (395.69,174.47) -- (420.56,174.47) ;
\draw    (395.69,211.88) -- (420.56,211.88) ;

\draw (350.2,159.4) node [anchor=north west][inner sep=0.75pt]  [font=\scriptsize]  {$\ e_{1}$};
\draw (394.2,214.4) node [anchor=north west][inner sep=0.75pt]  [font=\scriptsize]  {$\ e_{2}$};
\draw (433.2,160.4) node [anchor=north west][inner sep=0.75pt]  [font=\scriptsize]  {$\ e_{3}$};
\draw (434.71,212.89) node [anchor=north west][inner sep=0.75pt]  [font=\scriptsize]  {$\ e_{4}$};
\end{tikzpicture}
\caption{An example of an equilibrium assignment in a cylinder graph, when $n_k\ge 3$, with 5 empty nodes.}
\label{fig:equilibrium:cylinder}
\end{figure}
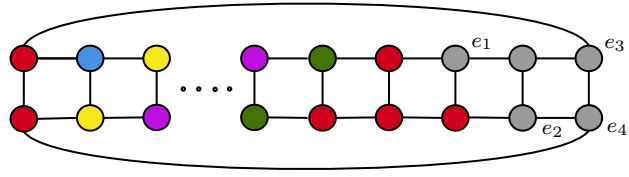

\medskip
\noindent
{\bf Case 3: $\mathbf{ n_k\geq 3}$.}  
First we select $\lfloor (|V|-n)/2 \rfloor$ consecutive columns of the graph; if $n$ is odd, we select an  additional node in the top row of an adjacent column: these nodes are denoted by $V_e$ and will be empty in our assignment (see the gray nodes in Figure~\ref{fig:equilibrium:cylinder}). Notice that at most four nodes in $V_e$  are adjacent to nodes outside $V_e$. Let $e_1$, $e_2$, $e_3$, and $e_4$ be the empty nodes of $V_e$ that are adjacent to the nodes of $V-V_e$, where, $e_1$, $e_2$ are in the top and bottom rows on one side of $V_e$ (not in the same column if $n$ is odd), and $e_3$, $e_4$ are in the top and bottom rows on the other side. Note that when $|V_e|=3$, $e_2=e_4$.

Create a sequence of ordered pairs of agents  as follows:  Add to the sequence  $n_1$ pairs from $T_1$ and $T_{2}$, followed by $n_2 - n_1$ pairs from the remaining agents in $T_{2}$ with $T_{3}$, and so on.
Stop the process when all unpaired agents are from $T_k$.

If $n$ is even, then the four empty nodes $e_1,e_2,e_3$ and $e_4$ will have exactly one non-empty node as a neighbor in any assignment, and therefore can only offer utility at most 1 to any agent. Since $n$ is even, the number of unpaired agents is also even. Create new ordered pairs from the unpaired agents of type $T_k$ and add these pairs to the end of the sequence of ordered pairs. Starting with the second next column from one end of $V_e$, place $n/2-1$ ordered pairs in the sequence until you reach the other end of $V_e$, and then place the last ordered pair in the last available column of $V-V_e$. Notice that this is guaranteed to be an equilibrium assignment, because either every agent has a utility at least 1, or every agent other than type $T_k$ has a utility at least 1 and all empty nodes are adjacent only to the agents of type $T_k$. 

Therefore, we assume that $n$ is odd and we are left with an odd number of unpaired agents. When $|V_e|>3$, $e_1$ is adjacent to 2 nodes of $V-V_e$, and $e_2,e_3$ and $e_4$ are adjacent to 1 node outside $V_e$. When $|V_e|=3$, $e_1$ and $e_2=e_4$ have 2 neighbors outside $V_e$, and $e_3$ has 1 neighbor outside $V_e$. Our assignments, depending on the cases below, will  ensure that  $e_1$, $e_2$ and $e_4$ are adjacent only to agents of type $T_k$. Now, we describe how to place the unpaired agents of type $T_k$ as well as the agents in the ordered pairs, depending on the cases below:

\medskip
\noindent 
{\bf Case 3~(a): 1 unpaired agents of type $T_k$.} 
Place the unpaired agent at the node below $e_1$. Since $n_k \geq 3$, the last 2 ordered pairs must contain an agent of type $T_k$; remove them from the sequence, and place each pair in the columns beside $e_1$ and $e_3$ such there $e_1$ and $e_4$ are adjacent only to agents of type $T_k$. Then, place all other ordered pairs of the sequence in the remaining columns of $V-V_e$. Notice that every agent has utility at least 1 and every node in $V_e$ can offer utility at most 1. 

\medskip
\noindent 
{\bf Case 3~(b): 3 unpaired agents of type $T_k$.} 
Place 2 unpaired agents of type $T_k$ as the two neighbors of  $e_1$; place the remaining unpaired agent beside $e_4$. Remove the first ordered pair from the sequence (both agents of this ordered pair are not of type $T_k$, since $k \ge 3$), split the agents in the ordered pair, and place them in the columns of $V-V_e$ that are half filled. Place all other ordered pairs in the remaining columns. Similar to the previous case, every agent has utility at least 1 and every node in $V_e$ can offer utility at most 1. 

\medskip
\noindent 
{\bf Case 3~(c): $5$ or more unpaired agents of type $T_k$.} 
Place 5 unpaired agents of type $T_k$ such that one is in $e_1$'s column and the other four in the two columns adjacent to $V_e$. Notice that all empty nodes are adjacent only to agents of type $T_k$. Now, place the agents of each ordered pair in different columns, and finally any remaining unpaired agents of type $T_k$ in any remaining positions in an arbitrary manner. Notice that every agent of type other than $T_k$ has utility at least 1 and none of the agents wants to jump to nodes of $V_e$. There may be agents of type $T_k$ that have utility 0, but they are not motivated to jump to empty nodes either.

In every case, we exhibited an equilibrium assignment, completing the proof. 
\end{proof}

Next we show that tori admit equilibrium assignments. To simplify the presentation, we assume that the torus has at least 9 rows and that the largest type has at least 8 agents. 

\begin{theorem} \label{thm:equilibrium:tori}
There exists an equilibrium assignment when the graph is an $m_1 \times m_2$ torus, where $m_1 \geq m_2 \geq 9$, and there are $n < m_1 m_2$ agents of $k \geq 3$ types, with the largest type containing $n_k \geq 8$ agents. 
\end{theorem}

\begin{proof} Consider an $m_1 \times m_2$ torus with $m_1 \ge m_2 \ge 9$. There are $k$ types, $T_1,T_2\ldots, T_k$. Recall that $n_j=|T_j|$ for $1 \le j \le k$, and we assume that $n_1 \leq n_2 \leq \ldots \leq n_k$. For simplicity, we also assume that $n_k\ge 8$.
We start by creating a sequence $L$ of agents of length $n$ as follows: Starting with $j=1$, alternate an agent from $T_j$ and one from $T_{j+1}$ until the agents in $T_j$ are finished; next alternate agents from $T_{j+2}$ and $T_{j+1}$, and so on (for example: $T_1,T_2,T_1,T_2,T_3,T_2,T_3,T_2,T_3,T_4,\ldots$), without using any agents of type $T_k$. Once the sequence is created, some agents of type $T_{k-1}$ might be remaining: we denote this set of agents as $X$. The sequence $L$ has the property that every agent is adjacent to agents of a different type; thus if the sequence is placed on the torus in such a way that every agent is a neighbor on the torus to at least one of its adjacent  agents in the sequence,  then every agent has utility at least 1. 

It remains to describe how to complete the sequence of agents with agents in $X$, as well as the agents of type $T_k$, and describe the assignment of agents to locations on the torus;  this depends on the number of agents $n$ as described below. The assignments we describe will satisfy either property \hyperref[itm:prop1]{\pOne} or \hyperref[itm:prop0]{\po}, which were also used in Theorem \ref{thm:equilibrium:trees}.

\medskip
\noindent
{\bf Case 1: $n \leq 5(m_2-2)$.}
Here, because the number of agents can fit in five rows of the torus, there are a lot of  empty nodes.  To complete the sequence $L$, we add agents in $X$ to the sequence $L$, and then  add agents of type $T_k$ between consecutive agents, starting from the end of the sequence. Since $n_k \geq n_{k-1}$, the neighbors of every agent in the sequence are of a type different from it,  and the last agent in $L$ is of type $T_k$.

We now partition the sequence $L$ as follows:  Let $i = |L| \mbox{ div } (m_2-2)$. We divide $L$ into a  subsequence $\hat{L}$ of length $i m_2$, and let $L'$ be the remaining subsequence of length $|L| \bmod m_2$. We now describe the placement of agents in $\hat{L}$, $L'$, and any remaining agents of type $T_k$ that were not in $L$.

\begin{itemize}
    \item Place the agents in $\hat{L}$ in the same $m_2-2$ columns in consecutive rows say $[r\ldots (r+i-1)]$, and place the agents in $L'$  in row $r + i + 2$. Since there are at least 9 rows, if $L'$ is empty, there are at least 4 empty rows between the last row of agents in $L'$ and the first row. If $L'$ is not empty, since agents in $\hat{L}$ occupy at most 4 rows, and there are at least 9 rows, there are at least two empty rows between the agents in $\hat{L}$ and agents in $L'$. 

\item If $n_k \leq  n_{k-1} + n_{k-2} +\ldots + n_1$, we are done, as the property \hyperref[itm:prop0]{\po} is satisfied. Otherwise, there are some remaining agents of type $T_k$ that were not in the sequence $L$. Place them one by one in empty nodes adjacent to agents in $L$ that are of a type different from $T_k$, where they can get a utility of 1. Start by placing them in the $i$ rows in which agents from $\hat{L}$ were placed, ensuring that when the first and last agent from $\hat{L}$ in a row are of the {\em same} type (that is different from $T_k$), we place agents of type $T_k$ as neighbors of both of these agents. (This is necessary, as otherwise, we would create an empty node which could offer utility 2.) If this is not possible, then it can only be because we have run out of agents of type $T_k$ to place, or we have only one such agent. In the latter case, we simply place that last agent of type $T_k$ in a row adjacent to the rows in which agents of type $\hat{L}$ are placed.

If we still have agents of type $T_k$ left to place, and there are no empty nodes adjacent to agents of other types,  place the remaining agents in any available empty nodes on the torus; they will have utility 0, but since all empty nodes are adjacent only to agents of type $T_k$, they are not  motivated to jump. So, \hyperref[itm:prop1]{\pOne} or \hyperref[itm:prop0]{\po} is satisfied.
\end{itemize}

\medskip
\noindent
{\bf Case 2: $5(m_2-2) < n \leq m_1m_2 - 2m_2$ }
We create a quasi-rectangular area $A$  of empty nodes, with $\lfloor(m_1m_2-n)/m_2\rfloor$ complete rows and a row that is not complete with $(m_1m_2-n)\mod m_2$ consecutive nodes. No agent will be assigned to the nodes of $A$. We call the nodes in $A$ that are adjacent to at least 2 nodes outside $A$ as \emph{corners} (see Figure~\ref{fig:equilibrium:torus}).  Notice that in $A$ there can be at most 2 \emph{corners} and they will be adjacent to at most 4 agents. 

\medskip
\noindent 
{\bf Case 2~(a):} If $n_k \ge 2m_2$, then first use $2m_2$ agents of type $T_k$ to enclose $A$. Next we describe how to complete the sequence $L$. Append to the sequence $L$ an alternating sequence of agents of type $T_{k}$ and $T_{k-1}$ until one of the types is exhausted. Denote the remaining agents by the set $Y$; these agents are either all of type $T_{k-1}$ or all of type $T_k$. 


\begin{itemize}
\item
If $Y$ has agents of type $T_{k-1}$, then we will first assign these agents. Notice that in this case $|Y| \leq 2m_2$, since $n_{k-1} \leq n_{k}$. We will assign all these agents \emph{consecutively} adjacent to the agents of type $T_k$ already assigned to the torus, and bordering the empty nodes. Thus all these agents in $Y$ will have a utility of at least 1. Then we assign the agents of $L$ to the available empty nodes of the torus in a snake-like pattern, starting with a row that has the least number of available empty nodes. Notice that property \hyperref[itm:prop1]{\pOne} is satisfied.
\item
If instead $Y$ has agents of type $T_k$, we first assign the agents of $L$ to the available empty nodes in a snake-like pattern, and then place all the agents of type $T_k$ that are not in $L$ in the remaining unassigned nodes. Here, property \hyperref[itm:prop0]{\po} is satisfied.
\end{itemize}

\medskip
\noindent 
{\bf Case 2~(b):} Otherwise, $n_k < 2m_2$, and we cannot ensure that empty nodes are adjacent to a single type $T_k$. So, we first place 4 agents of type $T_k$ adjacent to the two corners (recall that $n_k \geq 8$).  Then we put other agents of type $T_k$ consecutively on the border. Next we place agents in the set $X$ (of type $T_{k-1}$) adjacent to the agents of type $T_k$ that were just placed on the torus, ensuring that they all obtain utility 1. Finally, we place the sequence $L$ on the torus in snake-like order. See Figure \ref{fig:equilibrium:torus} for an example assignment. Notice that property \hyperref[itm:prop1]{\pOne} is satisfied in this case.


\medskip
\noindent
{\bf Case 3: $n > m_1m_2-2m_2$.} In this case, when the number of empty nodes are relatively lower, create a quasi-rectangular area $A$ that is located within 2 columns and the length of $A$ in the two columns differs by at most one. 

If agents of type $T_k$ can enclose the area $A$, then use them to enclose $A$. After that, follow the approach used in Case 2~(a) to complete the sequence $L$, and assign the remaining agents.

Otherwise, we first place agents of type $T_k$ next to the 4 corners; this is possible, since at most 8 agents are required, and $n_k \geq 8$. Next, follow  the approach used in Case 2~(b) to assign the remaining agents.

\noindent
In every case, we showed that there is an equilibrium assignment, completing the proof. 
\end{proof}

\begin{figure}[t]
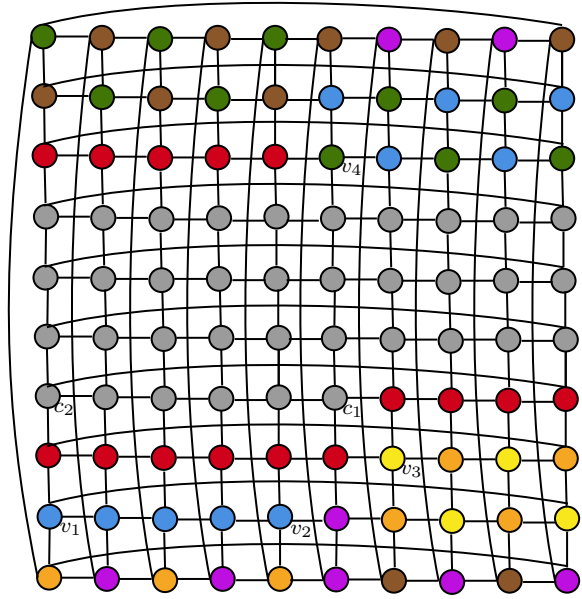

\centering

\tikzset{every picture/.style={line width=0.75pt}} 



\caption{An example of an equilibrium assignment for {\bf Case 2~(b)}. First the corners, $c_1$ and $c_2$, are covered with red agents, and the boundary is partially covered with the remaining red agents. Some blue agents are remaining after creating $L$, so they are placed adjacent to the red agents on the boundary, from nodes $v_1$ to $v_2$. Then, the agents of $L$ are then placed in a snake-like pattern starting at $v_3$ and ending at $v_4$.}
\label{fig:equilibrium:torus}
\end{figure}

\section{Quality of Equilibria}
In this section we focus on the diversity quality of equilibria as measured by the social welfare (the total utility of all agents) and the number of colorful edges between agents of different types. In particular, we show tight bounds on the price of anarchy with respect to both of these objective functions, and a lower bound on the price of stability. 

\subsection{Price of Anarchy: Social Welfare}

We start by showing an upper bound on the price of anarchy which holds for any class of graphs and is a function of the number of agents $n$ and the maximum cardinality over all types. 

\begin{theorem} \label{thm:poa:sw:general}
The price of anarchy with respect to the social welfare is exactly $\frac{n(k-1)}{n-\max_T n_T +1}$.
\end{theorem}

\begin{proof}
For the upper bound, consider an arbitrary equilibrium assignment $\bv$ and let $v$ be an empty node that is adjacent to an agent of some type $R$. All agents of type different than $R$ must have utility at least $1$ to not have incentive to jump to $v$. In addition, at least one agent of type $R$ must be adjacent to an agent of a different type; otherwise, since the graph is connected, there would be an empty node (possibly $v$) that is adjacent to an agent of type different than $R$ where agents of type $R$ would have incentive to jump. Therefore, 
\begin{align*}
    \SW(\bv) \geq 1 + \sum_{T \neq R} n_T  = n-n_R+1 \geq n-\max_T n_T +1.
\end{align*}
Since the optimal social welfare is at most $n(k-1)$, the price of anarchy is at most $\frac{n(k-1)}{n-\max_T n_T +1}$. 

For the lower bound, consider a game with $k \geq 2$ types $\{T_1, \ldots, T_k\}$. Let $R:=T_1$ be the type with the most agents.
The graph consists of the following components:
\begin{itemize}
\item A clique $K_n$ of size $n$;
\item A line $\ell_T$ of size $n_{T}$ for each type $T$.
\end{itemize}
The components are connected as follows: 
A node of $K_n$ is connected to a node of $\ell_R$. 
For each $t \in \{2,\ldots,k-1\}$, all but the last node of $\ell_{T_t}$ are connected to the last node of $\ell_{T_{t-1}}$.
All the nodes of $\ell_{T_k}$ are connected to the last node of $\ell_{T_{k-1}}$; See Figure~\ref{fig:poa:sw:general:lower} for an example for $k=3$. 

The optimal assignment is for the agents to occupy all the nodes of the $K_n$ component, which leads to a utility of $k-1$ for each of them, and a social welfare of $n(k-1)$. A different equilibrium assignment is the following: For each type $T$, the agents of type $T$ occupy the nodes of line $\ell_T$. Since all agents of type $T \neq R$ have utility $1$, none of them wants to jump to any of the empty nodes; in particular, there is only one node that is adjacent to an agent occupied by a single agent of type $R$. The social welfare of this equilibrium is $n-n_R+1$ (note that the agent of type $R$ that occupies the last node of $\ell_R$ also has utility $1$), and thus the price of anarchy is at least
$\frac{n(k-1)}{n-n_R +1}$.
\end{proof}

\begin{figure}[t]
\centering
\tikzset{every picture/.style={line width=0.75pt}} 

\begin{tikzpicture}[x=0.75pt,y=0.75pt,yscale=-1,xscale=1]

\draw   (75,77) -- (147,77) -- (147,149) -- (75,149) -- cycle ;
\draw    (110.5,139.5) -- (110.5,162) ;
\draw    (120.5,171) -- (133.5,171) ;
\draw    (153.5,171) -- (166.5,171) ;
\draw    (242.5,220) -- (275.5,242) ;
\draw    (242.5,220) -- (308.5,242) ;
\draw    (209.5,202) -- (176.5,180) ;
\draw  [fill={rgb, 255:red, 208; green, 2; blue, 27 }  ,fill opacity=1 ] (100.5,171) .. controls (100.5,166.03) and (104.98,162) .. (110.5,162) .. controls (116.02,162) and (120.5,166.03) .. (120.5,171) .. controls (120.5,175.97) and (116.02,180) .. (110.5,180) .. controls (104.98,180) and (100.5,175.97) .. (100.5,171) -- cycle ;
\draw  [fill={rgb, 255:red, 208; green, 2; blue, 27 }  ,fill opacity=1 ] (133.5,171) .. controls (133.5,166.03) and (137.98,162) .. (143.5,162) .. controls (149.02,162) and (153.5,166.03) .. (153.5,171) .. controls (153.5,175.97) and (149.02,180) .. (143.5,180) .. controls (137.98,180) and (133.5,175.97) .. (133.5,171) -- cycle ;
\draw  [fill={rgb, 255:red, 208; green, 2; blue, 27 }  ,fill opacity=1 ] (166.5,171) .. controls (166.5,166.03) and (170.98,162) .. (176.5,162) .. controls (182.02,162) and (186.5,166.03) .. (186.5,171) .. controls (186.5,175.97) and (182.02,180) .. (176.5,180) .. controls (170.98,180) and (166.5,175.97) .. (166.5,171) -- cycle ;
\draw  [fill={rgb, 255:red, 155; green, 155; blue, 155 }  ,fill opacity=1 ] (100.5,130.5) .. controls (100.5,125.53) and (104.98,121.5) .. (110.5,121.5) .. controls (116.02,121.5) and (120.5,125.53) .. (120.5,130.5) .. controls (120.5,135.47) and (116.02,139.5) .. (110.5,139.5) .. controls (104.98,139.5) and (100.5,135.47) .. (100.5,130.5) -- cycle ;
\draw    (176.5,180) -- (176.5,202) ;
\draw    (186.5,211) -- (199.5,211) ;
\draw    (219.5,211) -- (232.5,211) ;
\draw  [fill={rgb, 255:red, 74; green, 144; blue, 226 }  ,fill opacity=1 ] (166.5,211) .. controls (166.5,206.03) and (170.98,202) .. (176.5,202) .. controls (182.02,202) and (186.5,206.03) .. (186.5,211) .. controls (186.5,215.97) and (182.02,220) .. (176.5,220) .. controls (170.98,220) and (166.5,215.97) .. (166.5,211) -- cycle ;
\draw  [fill={rgb, 255:red, 74; green, 144; blue, 226 }  ,fill opacity=1 ] (199.5,211) .. controls (199.5,206.03) and (203.98,202) .. (209.5,202) .. controls (215.02,202) and (219.5,206.03) .. (219.5,211) .. controls (219.5,215.97) and (215.02,220) .. (209.5,220) .. controls (203.98,220) and (199.5,215.97) .. (199.5,211) -- cycle ;
\draw  [fill={rgb, 255:red, 74; green, 144; blue, 226 }  ,fill opacity=1 ] (232.5,211) .. controls (232.5,206.03) and (236.98,202) .. (242.5,202) .. controls (248.02,202) and (252.5,206.03) .. (252.5,211) .. controls (252.5,215.97) and (248.02,220) .. (242.5,220) .. controls (236.98,220) and (232.5,215.97) .. (232.5,211) -- cycle ;
\draw    (242.5,219.5) -- (242.5,242) ;
\draw    (252.5,251) -- (265.5,251) ;
\draw    (285.5,251) -- (298.5,251) ;
\draw  [fill={rgb, 255:red, 65; green, 117; blue, 5 }  ,fill opacity=1 ] (232.5,251) .. controls (232.5,246.03) and (236.98,242) .. (242.5,242) .. controls (248.02,242) and (252.5,246.03) .. (252.5,251) .. controls (252.5,255.97) and (248.02,260) .. (242.5,260) .. controls (236.98,260) and (232.5,255.97) .. (232.5,251) -- cycle ;
\draw  [fill={rgb, 255:red, 65; green, 117; blue, 5 }  ,fill opacity=1 ] (265.5,251) .. controls (265.5,246.03) and (269.98,242) .. (275.5,242) .. controls (281.02,242) and (285.5,246.03) .. (285.5,251) .. controls (285.5,255.97) and (281.02,260) .. (275.5,260) .. controls (269.98,260) and (265.5,255.97) .. (265.5,251) -- cycle ;
\draw  [fill={rgb, 255:red, 65; green, 117; blue, 5 }  ,fill opacity=1 ] (298.5,251) .. controls (298.5,246.03) and (302.98,242) .. (308.5,242) .. controls (314.02,242) and (318.5,246.03) .. (318.5,251) .. controls (318.5,255.97) and (314.02,260) .. (308.5,260) .. controls (302.98,260) and (298.5,255.97) .. (298.5,251) -- cycle ;

\draw (99,89.4) node [anchor=north west][inner sep=0.75pt]    {$K_{n} \ $};
\draw (72,162.4) node [anchor=north west][inner sep=0.75pt]    {$\ell _{R}$};
\draw (139,202.4) node [anchor=north west][inner sep=0.75pt]    {$\ell _{B}$};
\draw (205,242.4) node [anchor=north west][inner sep=0.75pt]    {$\ell _{G}$};
\end{tikzpicture}
\caption{The graph of the game considered in the proof of the lower bound in Theorem~\ref{thm:poa:sw:general} for $k=3$ types $R, B, G$. The depicted assignment (according to which the agents occupy the nodes of the lines $\ell_R, \ell_B$ and $\ell_G$) is the equilibrium with minimum social welfare that leads to the desired price of anarchy lower bound.}
\label{fig:poa:sw:general:lower}
\end{figure}
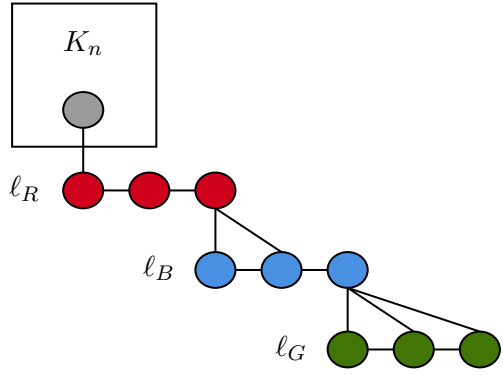

Using Theorem~\ref{thm:poa:sw:general}, we can derive many asymptotically tight bounds on the price of anarchy with respect to the social welfare, for example by considering games in which the types might be asymmetric or symmetric.  

\begin{corollary}
The price of anarchy with respect to the social welfare is $\Theta(n)$ in general, and $\Theta(k)$ for symmetric types. 
\end{corollary}

\begin{proof}
For asymmetric types, since each of the $k$ types has cardinality at least $1$, $\max_T n_T \leq n-k+1$, and thus the price of anarchy is at most 
$$\frac{n(k-1)}{n-\max_T n_T+1} \leq n \cdot \frac{k-1}{k} \leq n.$$
By setting the $n_R=n-k+1$ and $n_T=1$ for $T\neq R$ in the lower bound construction of Theorem~\ref{thm:poa:sw:general}, we get an asymptotically tight lower bound that is linear in the number of agents.

For symmetric types, $n_T = n/k$ for each type $T$, and thus the price of anarchy is at most
$$\frac{n(k-1)}{n(1-1/k) +1} = k\cdot \frac{n(k-1)}{n(k-1) + k} \leq k.$$
By setting $n_T = n/k$ for each type $T$ in the lower bound construction of Theorem~\ref{thm:poa:sw:general}, we get an asymptotically tight lower bound that is linear in the number of types.
\end{proof}

The previous statements show that the price of anarchy can be high in general graphs, even when the types are symmetric.
However, in simple graphs, such as lines and cycles, the price of anarchy is much smaller. 

\begin{theorem}\label{thm:poa:sw:cycle}
When the types are symmetric and the graph is of degree at most $2$, the price of anarchy with respect to the social welfare is $4/3$ for $k=2$ and $\frac{2k}{k-1}$ for $k \geq 3$.
\end{theorem}

\begin{proof}
We start with the case of $k=2$ types, red and blue. For the upper bound, consider an arbitrary equilibrium assignment $\bv$ and let $v$ be an empty node that is adjacent to a red agent. Then, all blue agents must have utility $1$ to not have incentive to jump to $v$. Since each node has degree at most $2$, this further means that for every two blue agents, there is at least one red agent with utility $1$ as well. Hence, the social welfare of $\bv$ is at least $n/2 + n/4 = 3n/4$. Since the maximum possible social welfare is $n$, the price of anarchy is at most $4/3$. 

For the lower bound, consider a cycle or a line with $n+1$ nodes. Clearly, the assignment according to which agents of alternating types occupy adjacent nodes is the one with optimal social welfare of $n$. A different equilibrium assignment has the following pattern: There are $n/4$ triplets consisting of a red and two blue agents (in this order) that occupy consecutive nodes. The remaining red agents occupy the nodes after the ones occupied by all the agents in these triplets. See Figure~\ref{fig:poa:sw:cycle:2-types} for an example of this assignment with a few agents. 
Hence, there are exactly $3n/4 + 1$ agents with utility $1$, leading to a price of anarchy lower bound of approximately $4/3$ for a sufficiently large $n$. 

We now consider the case $k \geq 3$. For the upper bound, if at equilibrium there is an empty node $v$ that is adjacent to a red agent, then all non-red agents must have utility $1$. Hence, the social welfare at equilibrium is at least $n - n/k = n(k-1)/k$. Since the nodes of the graph have degree (at most) $2$, the maximum social welfare is $2n$, leading to a price of anarchy of at most $2k/(k-1)$. 

For the lower bound, again consider a cycle or a line with $n+1$ nodes. The assignment according to which agents of alternating types (in a round-robin manner) occupy adjacent nodes is the one with optimal social welfare of $2n$. A different equilibrium assignment has the following pattern: Consider any ordering of the non-red types. The first node is occupied by a red agent. Then, two agents of the same type are followed by two agents of the type that comes next according to the ordering repeatedly. The remaining red agents occupy the nodes after all the non-red agents. See Figure~\ref{fig:poa:sw:cycle:ktypes} for an example of this assignment with a few agents. Hence, there are $n-n/k+2 = n(k-1)/k+2$ agents with utility $1$, leading to a price of anarchy lower bound that approaches $2k/(k-1)$ for a sufficiently large $n$. 
\end{proof}

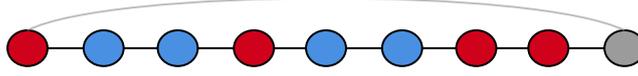
\begin{figure}[t]
\centering
\tikzset{every picture/.style={line width=0.75pt}} 
\begin{tikzpicture}[x=0.75pt,y=0.75pt,yscale=-1,xscale=1]

\draw  [fill={rgb, 255:red, 208; green, 2; blue, 27 }  ,fill opacity=1 ] (126,94) .. controls (126,89.03) and (130.48,85) .. (136,85) .. controls (141.52,85) and (146,89.03) .. (146,94) .. controls (146,98.97) and (141.52,103) .. (136,103) .. controls (130.48,103) and (126,98.97) .. (126,94) -- cycle ;
\draw  [fill={rgb, 255:red, 74; green, 144; blue, 226 }  ,fill opacity=1 ] (164,94) .. controls (164,89.03) and (168.48,85) .. (174,85) .. controls (179.52,85) and (184,89.03) .. (184,94) .. controls (184,98.97) and (179.52,103) .. (174,103) .. controls (168.48,103) and (164,98.97) .. (164,94) -- cycle ;
\draw    (146,94) -- (164,94) ;
\draw  [fill={rgb, 255:red, 74; green, 144; blue, 226 }  ,fill opacity=1 ] (201,94) .. controls (201,89.03) and (205.48,85) .. (211,85) .. controls (216.52,85) and (221,89.03) .. (221,94) .. controls (221,98.97) and (216.52,103) .. (211,103) .. controls (205.48,103) and (201,98.97) .. (201,94) -- cycle ;
\draw  [fill={rgb, 255:red, 208; green, 2; blue, 27 }  ,fill opacity=1 ] (239,94) .. controls (239,89.03) and (243.48,85) .. (249,85) .. controls (254.52,85) and (259,89.03) .. (259,94) .. controls (259,98.97) and (254.52,103) .. (249,103) .. controls (243.48,103) and (239,98.97) .. (239,94) -- cycle ;
\draw    (221,94) -- (239,94) ;
\draw    (184,94) -- (201,94) ;
\draw  [fill={rgb, 255:red, 74; green, 144; blue, 226 }  ,fill opacity=1 ] (275,94) .. controls (275,89.03) and (279.48,85) .. (285,85) .. controls (290.52,85) and (295,89.03) .. (295,94) .. controls (295,98.97) and (290.52,103) .. (285,103) .. controls (279.48,103) and (275,98.97) .. (275,94) -- cycle ;
\draw  [fill={rgb, 255:red, 74; green, 144; blue, 226 }  ,fill opacity=1 ] (313,94) .. controls (313,89.03) and (317.48,85) .. (323,85) .. controls (328.52,85) and (333,89.03) .. (333,94) .. controls (333,98.97) and (328.52,103) .. (323,103) .. controls (317.48,103) and (313,98.97) .. (313,94) -- cycle ;
\draw    (295,94) -- (313,94) ;
\draw    (259,94) -- (275,94) ;
\draw  [fill={rgb, 255:red, 208; green, 2; blue, 27 }  ,fill opacity=1 ] (350,94) .. controls (350,89.03) and (354.48,85) .. (360,85) .. controls (365.52,85) and (370,89.03) .. (370,94) .. controls (370,98.97) and (365.52,103) .. (360,103) .. controls (354.48,103) and (350,98.97) .. (350,94) -- cycle ;
\draw    (333,94) -- (350,94) ;
\draw  [fill={rgb, 255:red, 208; green, 2; blue, 27 }  ,fill opacity=1 ] (386,94) .. controls (386,89.03) and (390.48,85) .. (396,85) .. controls (401.52,85) and (406,89.03) .. (406,94) .. controls (406,98.97) and (401.52,103) .. (396,103) .. controls (390.48,103) and (386,98.97) .. (386,94) -- cycle ;
\draw  [fill={rgb, 255:red, 155; green, 155; blue, 155 }  ,fill opacity=1 ] (424,94) .. controls (424,89.03) and (428.48,85) .. (434,85) .. controls (439.52,85) and (444,89.03) .. (444,94) .. controls (444,98.97) and (439.52,103) .. (434,103) .. controls (428.48,103) and (424,98.97) .. (424,94) -- cycle ;
\draw    (406,94) -- (424,94) ;
\draw    (370,94) -- (386,94) ;
\draw [color={rgb, 255:red, 155; green, 155; blue, 155 }  ,draw opacity=0.65 ]   (136,85) .. controls (177.33,64) and (388.33,63) .. (434,85) ;
\end{tikzpicture}
\caption{An example of the assignment that gives the price of anarchy lower bound of $4/3$ with respect to the social welfare for $k=2$ on lines and cycles.}
\label{fig:poa:sw:cycle:2-types}
\end{figure}

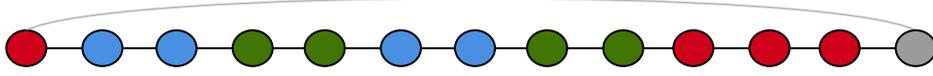
\begin{figure}[t]
    \centering

\tikzset{every picture/.style={line width=0.75pt}} 

\begin{tikzpicture}[x=0.75pt,y=0.75pt,yscale=-1,xscale=1]

\draw  [fill={rgb, 255:red, 208; green, 2; blue, 27 }  ,fill opacity=1 ] (55,94) .. controls (55,89.03) and (59.48,85) .. (65,85) .. controls (70.52,85) and (75,89.03) .. (75,94) .. controls (75,98.97) and (70.52,103) .. (65,103) .. controls (59.48,103) and (55,98.97) .. (55,94) -- cycle ;
\draw  [fill={rgb, 255:red, 74; green, 144; blue, 226 }  ,fill opacity=1 ] (93,94) .. controls (93,89.03) and (97.48,85) .. (103,85) .. controls (108.52,85) and (113,89.03) .. (113,94) .. controls (113,98.97) and (108.52,103) .. (103,103) .. controls (97.48,103) and (93,98.97) .. (93,94) -- cycle ;
\draw    (75,94) -- (93,94) ;
\draw  [fill={rgb, 255:red, 74; green, 144; blue, 226 }  ,fill opacity=1 ] (130,94) .. controls (130,89.03) and (134.48,85) .. (140,85) .. controls (145.52,85) and (150,89.03) .. (150,94) .. controls (150,98.97) and (145.52,103) .. (140,103) .. controls (134.48,103) and (130,98.97) .. (130,94) -- cycle ;
\draw  [fill={rgb, 255:red, 65; green, 117; blue, 5 }  ,fill opacity=1 ] (168,94) .. controls (168,89.03) and (172.48,85) .. (178,85) .. controls (183.52,85) and (188,89.03) .. (188,94) .. controls (188,98.97) and (183.52,103) .. (178,103) .. controls (172.48,103) and (168,98.97) .. (168,94) -- cycle ;
\draw    (150,94) -- (168,94) ;
\draw    (113,94) -- (130,94) ;
\draw  [fill={rgb, 255:red, 65; green, 117; blue, 5 }  ,fill opacity=1 ] (204,94) .. controls (204,89.03) and (208.48,85) .. (214,85) .. controls (219.52,85) and (224,89.03) .. (224,94) .. controls (224,98.97) and (219.52,103) .. (214,103) .. controls (208.48,103) and (204,98.97) .. (204,94) -- cycle ;
\draw  [fill={rgb, 255:red, 74; green, 144; blue, 226 }  ,fill opacity=1 ] (242,94) .. controls (242,89.03) and (246.48,85) .. (252,85) .. controls (257.52,85) and (262,89.03) .. (262,94) .. controls (262,98.97) and (257.52,103) .. (252,103) .. controls (246.48,103) and (242,98.97) .. (242,94) -- cycle ;
\draw    (224,94) -- (242,94) ;
\draw    (188,94) -- (204,94) ;
\draw  [fill={rgb, 255:red, 74; green, 144; blue, 226 }  ,fill opacity=1 ] (279,94) .. controls (279,89.03) and (283.48,85) .. (289,85) .. controls (294.52,85) and (299,89.03) .. (299,94) .. controls (299,98.97) and (294.52,103) .. (289,103) .. controls (283.48,103) and (279,98.97) .. (279,94) -- cycle ;
\draw    (262,94) -- (279,94) ;
\draw  [fill={rgb, 255:red, 65; green, 117; blue, 5 }  ,fill opacity=1 ] (315,94) .. controls (315,89.03) and (319.48,85) .. (325,85) .. controls (330.52,85) and (335,89.03) .. (335,94) .. controls (335,98.97) and (330.52,103) .. (325,103) .. controls (319.48,103) and (315,98.97) .. (315,94) -- cycle ;
\draw  [fill={rgb, 255:red, 65; green, 117; blue, 5 }  ,fill opacity=1 ] (353,94) .. controls (353,89.03) and (357.48,85) .. (363,85) .. controls (368.52,85) and (373,89.03) .. (373,94) .. controls (373,98.97) and (368.52,103) .. (363,103) .. controls (357.48,103) and (353,98.97) .. (353,94) -- cycle ;
\draw    (335,94) -- (353,94) ;
\draw    (299,94) -- (315,94) ;
\draw [color={rgb, 255:red, 155; green, 155; blue, 155 }  ,draw opacity=0.65 ]   (65,85) .. controls (106.33,64) and (463.33,63) .. (509,85) ;
\draw  [fill={rgb, 255:red, 208; green, 2; blue, 27 }  ,fill opacity=1 ] (388,94) .. controls (388,89.03) and (392.48,85) .. (398,85) .. controls (403.52,85) and (408,89.03) .. (408,94) .. controls (408,98.97) and (403.52,103) .. (398,103) .. controls (392.48,103) and (388,98.97) .. (388,94) -- cycle ;
\draw  [fill={rgb, 255:red, 208; green, 2; blue, 27 }  ,fill opacity=1 ] (426,94) .. controls (426,89.03) and (430.48,85) .. (436,85) .. controls (441.52,85) and (446,89.03) .. (446,94) .. controls (446,98.97) and (441.52,103) .. (436,103) .. controls (430.48,103) and (426,98.97) .. (426,94) -- cycle ;
\draw    (408,94) -- (426,94) ;
\draw    (373,94) -- (388,94) ;
\draw  [fill={rgb, 255:red, 208; green, 2; blue, 27 }  ,fill opacity=1 ] (461,94) .. controls (461,89.03) and (465.48,85) .. (471,85) .. controls (476.52,85) and (481,89.03) .. (481,94) .. controls (481,98.97) and (476.52,103) .. (471,103) .. controls (465.48,103) and (461,98.97) .. (461,94) -- cycle ;
\draw  [fill={rgb, 255:red, 155; green, 155; blue, 155 }  ,fill opacity=1 ] (499,94) .. controls (499,89.03) and (503.48,85) .. (509,85) .. controls (514.52,85) and (519,89.03) .. (519,94) .. controls (519,98.97) and (514.52,103) .. (509,103) .. controls (503.48,103) and (499,98.97) .. (499,94) -- cycle ;
\draw    (481,94) -- (499,94) ;
\draw    (446,94) -- (461,94) ;
\end{tikzpicture}

    \caption{An example of the assignment that gives the price of anarchy lower bound of $2k/(k-1)$ with respect to the social welfare for $k\geq 3$ on lines and cycles.}
    \label{fig:poa:sw:cycle:ktypes}
\end{figure}

\subsection{Price of Anarchy: Colorful Edges}
We now turn our attention to the second objective, the number of colorful edges. We first show a lower bound on the number of colorful edges achieved in any equilibrium; this lemma will be used in several proofs in the following for showing upper bounds on the price of anarchy. 

\begin{lemma} \label{lem:poa:colorful:equilibrium}
The number of colorful edges in any equilibrium is at least $(n-\max_T n_T)/2$.
\end{lemma}

\begin{proof}
Consider an arbitrary equilibrium assignment $\bv$ and let $v$ be an empty node that is adjacent to an agent of some type $R$. All agents of type different than $R$ must have at least one neighbor of different type than their own to not have incentive to jump to $v$. Since a colorful edge connects two agents of different type, the number of colorful edges in $\bv$ is at least $(n-n_R)/2 \geq (n-\max_T n_T)/2$. 
\end{proof}

Now we are ready to show that the price of anarchy with respect to the number of colorful edges is linear in the number of agents, even when the types have the same exact cardinality.

\begin{theorem} \label{thm:poa:colorful:general}
The price of anarchy with respect to the number of colorful edges is $\Theta(n)$, even when the types are symmetric. 
\end{theorem}

\begin{proof}
For the upper bound, observe that the maximum number of colorful edges would be achieved if all agents were to be connected to all agents of different type than their own. Hence, the optimal number of colorful edges is at most the number of edges between all agents minus the number of edges between all agents of the same type. Since $n = \sum_T n_T > \max_T n_T$, we have the following upper bound:
\begin{align*}
    \frac12 \bigg( n(n-1) - \sum_T n_T(n_T-1) \bigg) 
    &= \frac12 \bigg( n^2 - \sum_T n_T^2 \bigg) \\
    & \leq \frac12 ( n^2 - (\max_T n_T)^2 )\\
    &= \frac12 (n-\max_T n_T)(n+ \max_T n_T).
\end{align*}
So, by Lemma~\ref{lem:poa:colorful:equilibrium}, the price of anarchy is at most $n+\max_T n_T \leq 2n$. 

For the lower bound, consider a game with $k \geq 2$ types $\{T_1, \dots, T_k\}$ consisting of $n/k$ agents each, and assume that $n/k$ is an even number. The graph consists of the following components:
\begin{itemize}
    \item A clique $K_n$ of size $n$;
    \item A cycle $c_n$ of size $n$.
\end{itemize}
The components are connected as follows: A node of $K_n$ is connected to a node of $c_n$. See Figure~\ref{fig:poa:colorful:general:lower} for an example for $k=3$.

The optimal assignment is for the agents to occupy all the nodes of the $K_n$ component, which leads to the maximum possible number of colorful edges equal to 
\begin{align*}
    \frac{1}{2}\bigg( n(n-1) - k \cdot \frac{n}{k}\left(\frac{n}{k}-1 \right) \bigg) = \frac{n^2}{2}\cdot \frac{k-1}{k} \geq \frac{n^2}{4}.
\end{align*}
Now consider the following equilibrium assignment according to which the agents occupy the nodes of the $c_n$ component such that each agent has a neighbor of the same type and a neighbors of different type. As all agents have utility $1$, no agent has incentive to jump to the empty node of $K_n$ that is adjacent to a node of $c_n$. The number of colorful edges of this equilibrium is exactly equal to $n/2$, leading to a price of anarchy lower bound of $n/2$.
\end{proof}

\begin{figure}[t]
    \centering

\tikzset{every picture/.style={line width=0.75pt}} 

\begin{tikzpicture}[x=0.75pt,y=0.75pt,yscale=-1,xscale=1]

\draw   (100,72) -- (172,72) -- (172,142) -- (100,142) -- cycle ;
\draw  [fill={rgb, 255:red, 155; green, 155; blue, 155 }  ,fill opacity=1 ] (126,127) .. controls (126,122.03) and (130.48,118) .. (136,118) .. controls (141.52,118) and (146,122.03) .. (146,127) .. controls (146,131.97) and (141.52,136) .. (136,136) .. controls (130.48,136) and (126,131.97) .. (126,127) -- cycle ;
\draw  [fill={rgb, 255:red, 208; green, 2; blue, 27 }  ,fill opacity=1 ] (126,163) .. controls (126,158.03) and (130.48,154) .. (136,154) .. controls (141.52,154) and (146,158.03) .. (146,163) .. controls (146,167.97) and (141.52,172) .. (136,172) .. controls (130.48,172) and (126,167.97) .. (126,163) -- cycle ;
\draw  [fill={rgb, 255:red, 208; green, 2; blue, 27 }  ,fill opacity=1 ] (164,163) .. controls (164,158.03) and (168.48,154) .. (174,154) .. controls (179.52,154) and (184,158.03) .. (184,163) .. controls (184,167.97) and (179.52,172) .. (174,172) .. controls (168.48,172) and (164,167.97) .. (164,163) -- cycle ;
\draw  [fill={rgb, 255:red, 74; green, 144; blue, 226 }  ,fill opacity=1 ] (164,217) .. controls (164,212.03) and (168.48,208) .. (174,208) .. controls (179.52,208) and (184,212.03) .. (184,217) .. controls (184,221.97) and (179.52,226) .. (174,226) .. controls (168.48,226) and (164,221.97) .. (164,217) -- cycle ;
\draw  [fill={rgb, 255:red, 74; green, 144; blue, 226 }  ,fill opacity=1 ] (187,191) .. controls (187,186.03) and (191.48,182) .. (197,182) .. controls (202.52,182) and (207,186.03) .. (207,191) .. controls (207,195.97) and (202.52,200) .. (197,200) .. controls (191.48,200) and (187,195.97) .. (187,191) -- cycle ;
\draw  [fill={rgb, 255:red, 65; green, 117; blue, 5 }  ,fill opacity=1 ] (103,191) .. controls (103,186.03) and (107.48,182) .. (113,182) .. controls (118.52,182) and (123,186.03) .. (123,191) .. controls (123,195.97) and (118.52,200) .. (113,200) .. controls (107.48,200) and (103,195.97) .. (103,191) -- cycle ;
\draw  [fill={rgb, 255:red, 65; green, 117; blue, 5 }  ,fill opacity=1 ] (126,217) .. controls (126,212.03) and (130.48,208) .. (136,208) .. controls (141.52,208) and (146,212.03) .. (146,217) .. controls (146,221.97) and (141.52,226) .. (136,226) .. controls (130.48,226) and (126,221.97) .. (126,217) -- cycle ;
\draw    (146,163) -- (164,163) ;
\draw    (146,217) -- (164,217) ;
\draw    (126,163) -- (113,182) ;
\draw    (184,163) -- (197,182) ;
\draw    (113,200) -- (126,217) ;
\draw    (197,200) -- (184,217) ;
\draw    (136,136) -- (136,154) ;

\draw (125,80.4) node [anchor=north west][inner sep=0.75pt]    {$K_{n}$};
\draw (217,182.4) node [anchor=north west][inner sep=0.75pt]    {$c_{n}$};
\end{tikzpicture}

    \caption{The graph of the game considered in the proof of the lower bound in Theorem~\ref{thm:poa:colorful:general} for $k=3$ types (red, green, and blue) consisting of $2$ agents each.}
    \label{fig:poa:colorful:general:lower}
\end{figure}
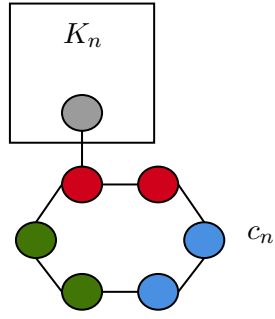

We next focus again on the case of symmetric types and specific classes of graphs including lines, cycles, and also regular graphs. 

\begin{theorem} \label{thm:poa:colorful:cycles}
When the types are symmetric and the graph is of degree at most $2$, the price of anarchy with respect to the number of colorful edges is $2$ when $k=2$ and $\frac{2k}{k-1}$ when $k \geq 3$.
\end{theorem}

\begin{proof}
We start with the case $k=2$. By Lemma~\ref{lem:poa:colorful:equilibrium}, since $\max_T n_T = n/2$, the number of colorful edges in any equilibrium is at least $n/2$. Since the nodes of the graph have degree (at most) $2$, each agent can at most participate in two colorful edges, and thus the maximum number of colorful edges is $n$. So, the price of anarchy is at most $2$. For the lower bound, consider again the game of Theorem~\ref{thm:poa:sw:cycle} for $k=2$ (see also Figure~\ref{fig:poa:sw:cycle:2-types}). It is not hard to observe that the equilibrium has $n/2$ colorful edges, whereas the optimal number of colorful edges is $n$, and thus the price of anarchy is at least $2$. 

We now turn to the case $k \geq 3$. By Lemma~\ref{lem:poa:colorful:equilibrium}, since $\max_T n_T = n/k$, there are at least $(n-n/k)/2 = n\frac{k-1}{2k}$ colorful edges at any equilibrium. Since the maximum number of colorful edges is $n$, the price of anarchy is then at most $\frac{2k}{k-1}$. For the lower bound, consider again the game of Theorem~\ref{thm:poa:sw:cycle} for $k\geq 3$ (see also Figure~\ref{fig:poa:sw:cycle:ktypes}). Each non-red agent participates in exactly one colorful edge, which implies that there are $n\frac{k-1}{2k}$ colorful edges at equilibrium, and the price of anarchy is at least $2k/(k-1)$. 
\end{proof}

\begin{theorem} \label{thm:poa:colorful:regular}
When the types are symmetric and the graph is $\delta$-regular, the price of anarchy with respect to the number of colorful edges is $\Theta(\delta)$.
\end{theorem}

\begin{proof}
For the upper bound, observe that the maximum number of colorful edges is $n\delta/2$ since each agent can have at most $\delta$ neighbors and can thus participate in at most $\delta$ colorful edges. Also, by Lemma~\ref{lem:poa:colorful:equilibrium}, since $\max_T n_T = n/k$, there are at least $\frac{n}{2}\frac{k-1}{k} \geq n/4$ colorful edges in an equilibrium, thus leading to an upper bound of $2\delta$ on the price of anarchy. 

For the lower bound, consider a game with $k \geq 2$ symmetric types $\{T_1, \ldots, T_k\}$, and $\delta = n/k+1$. The graph consists of the following components:
\begin{itemize}
    \item A cycle $c$ of size $k$ consisting of nodes $\{v_1, \ldots, v_k\}$;
    \item $k$ cliques $(K_\ell)_{\ell\in [k]}$ of size $n/k+1$ each. 
\end{itemize}
The components are connected to each other as follows: For each $\ell \in [k]$, node $v_\ell$ is connected to all but two nodes of clique $K_{(\ell+1) \bmod k}$. Also, the nodes in the cliques that are not connected to nodes of the cycle $c$ form another cycle (so that each of these nodes has one more neighbor in some other clique). Observe that this is indeed a $\delta$-regular graph: Each node $v_\ell$ of the cycle $c$ has degree $2 + n/k-1 = \delta$, and each node of a clique $K_\ell$ has degree $1 + n/k = \delta$. 

An equilibrium is the following: For each $\ell \in [k]$, node $v_\ell$ is occupied by an agent of type $T_\ell$ and the $n/k-1$ nodes of clique $K_{(\ell+1) \bmod k}$ that are adjacent to $v_\ell$ are occupied by agents of type $T_{(\ell+1) \bmod k}$; the remaining two nodes of each clique are left empty. See also Figure~\ref{fig:poa:colorful:regular} for an example with $k=3$. This is an equilibrium assignment as all agents have utility at least $1$ and each empty node is adjacent to only agents of one type. The only colorful edges in this assignment are those connecting the nodes of $c$ to each other and to nodes in the cliques. So, there are $k + k (\delta-2) = k(\delta-1)$ colorful edges in this equilibrium. On the other hand, in the optimal assignment, we can assign to each clique, one agent of every type, so that each agent has exactly $k-1$ neighbors, all of different type, leading to $k \frac{\delta(\delta-1)}{2}$ colorful edges. So, the price of anarchy is $\Omega(\delta)$.
\end{proof}

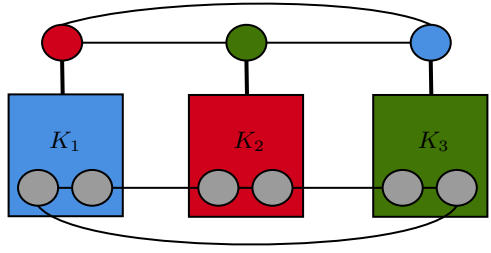
\begin{figure}[t]
    \centering

\tikzset{every picture/.style={line width=0.75pt}} 

\begin{tikzpicture}[x=0.75pt,y=0.75pt,yscale=-1,xscale=1]

\draw  [fill={rgb, 255:red, 74; green, 144; blue, 226 }  ,fill opacity=1 ] (113.29,147) -- (170.29,147) -- (170.29,208.29) -- (113.29,208.29) -- cycle ;
\draw  [fill={rgb, 255:red, 208; green, 2; blue, 27 }  ,fill opacity=1 ] (130,121) .. controls (130,116.03) and (134.48,112) .. (140,112) .. controls (145.52,112) and (150,116.03) .. (150,121) .. controls (150,125.97) and (145.52,130) .. (140,130) .. controls (134.48,130) and (130,125.97) .. (130,121) -- cycle ;
\draw  [fill={rgb, 255:red, 65; green, 117; blue, 5 }  ,fill opacity=1 ] (222,121) .. controls (222,116.03) and (226.48,112) .. (232,112) .. controls (237.52,112) and (242,116.03) .. (242,121) .. controls (242,125.97) and (237.52,130) .. (232,130) .. controls (226.48,130) and (222,125.97) .. (222,121) -- cycle ;
\draw    (150,121) -- (222,121) ;
\draw  [fill={rgb, 255:red, 74; green, 144; blue, 226 }  ,fill opacity=1 ] (314,121) .. controls (314,116.03) and (318.48,112) .. (324,112) .. controls (329.52,112) and (334,116.03) .. (334,121) .. controls (334,125.97) and (329.52,130) .. (324,130) .. controls (318.48,130) and (314,125.97) .. (314,121) -- cycle ;
\draw    (140,112) .. controls (183.29,98.29) and (294.29,97.29) .. (324,112) ;
\draw    (242,121) -- (314,121) ;
\draw [line width=1.5]    (140,130) -- (140.29,147.29) ;
\draw [line width=1.5]    (232,130) -- (232.29,147.29) ;
\draw [line width=1.5]    (324,130) -- (324.29,147.29) ;
\draw  [fill={rgb, 255:red, 208; green, 2; blue, 27 }  ,fill opacity=1 ] (203.29,147.29) -- (260.29,147.29) -- (260.29,208.57) -- (203.29,208.57) -- cycle ;
\draw  [fill={rgb, 255:red, 65; green, 117; blue, 5 }  ,fill opacity=1 ] (295.29,147.29) -- (352.29,147.29) -- (352.29,208.57) -- (295.29,208.57) -- cycle ;
\draw  [fill={rgb, 255:red, 155; green, 155; blue, 155 }  ,fill opacity=1 ] (145,194) .. controls (145,189.03) and (149.48,185) .. (155,185) .. controls (160.52,185) and (165,189.03) .. (165,194) .. controls (165,198.97) and (160.52,203) .. (155,203) .. controls (149.48,203) and (145,198.97) .. (145,194) -- cycle ;
\draw  [fill={rgb, 255:red, 155; green, 155; blue, 155 }  ,fill opacity=1 ] (118,194) .. controls (118,189.03) and (122.48,185) .. (128,185) .. controls (133.52,185) and (138,189.03) .. (138,194) .. controls (138,198.97) and (133.52,203) .. (128,203) .. controls (122.48,203) and (118,198.97) .. (118,194) -- cycle ;
\draw  [fill={rgb, 255:red, 155; green, 155; blue, 155 }  ,fill opacity=1 ] (208,194) .. controls (208,189.03) and (212.48,185) .. (218,185) .. controls (223.52,185) and (228,189.03) .. (228,194) .. controls (228,198.97) and (223.52,203) .. (218,203) .. controls (212.48,203) and (208,198.97) .. (208,194) -- cycle ;
\draw  [fill={rgb, 255:red, 155; green, 155; blue, 155 }  ,fill opacity=1 ] (235,194) .. controls (235,189.03) and (239.48,185) .. (245,185) .. controls (250.52,185) and (255,189.03) .. (255,194) .. controls (255,198.97) and (250.52,203) .. (245,203) .. controls (239.48,203) and (235,198.97) .. (235,194) -- cycle ;
\draw  [fill={rgb, 255:red, 155; green, 155; blue, 155 }  ,fill opacity=1 ] (300,194) .. controls (300,189.03) and (304.48,185) .. (310,185) .. controls (315.52,185) and (320,189.03) .. (320,194) .. controls (320,198.97) and (315.52,203) .. (310,203) .. controls (304.48,203) and (300,198.97) .. (300,194) -- cycle ;
\draw  [fill={rgb, 255:red, 155; green, 155; blue, 155 }  ,fill opacity=1 ] (327,194) .. controls (327,189.03) and (331.48,185) .. (337,185) .. controls (342.52,185) and (347,189.03) .. (347,194) .. controls (347,198.97) and (342.52,203) .. (337,203) .. controls (331.48,203) and (327,198.97) .. (327,194) -- cycle ;
\draw    (165,194) -- (208,194) ;
\draw    (255,194) -- (300,194) ;
\draw    (128,203) .. controls (148.29,229) and (321.29,229) .. (337,203) ;
\draw    (138,194) -- (145,194) ;
\draw    (228,194) -- (235,194) ;
\draw    (320,194) -- (327,194) ;

\draw (132,164.4) node [anchor=north west][inner sep=0.75pt]  [font=\footnotesize]  {$K_{1}$};
\draw (224,164.4) node [anchor=north west][inner sep=0.75pt]  [font=\footnotesize]  {$K_{2}$};
\draw (316,164.4) node [anchor=north west][inner sep=0.75pt]  [font=\footnotesize]  {$K_{3}$};
\end{tikzpicture}
    \caption{The $\delta$-regular graph considered in the proof of the lower bound in Theorem~\ref{thm:poa:colorful:regular} for $k=3$ types (red, green, and blue). The depicted assignment is the equilibrium minimizing the number of colorful edges.}
    \label{fig:poa:colorful:regular}
\end{figure}

\subsection{Price of Stability}
We conclude with a lower bound of approximately $3/2$ on the price of stability for both objectives. This implies that there are games in which the optimal assignment is not an equilibrium (in contrast to the games considered in the lower bounds on the price of anarchy). 

\begin{theorem} \label{thm:pos:sw:lower}
The price of stability with respect to the social welfare and the number of colorful edges is at least $3/2-\varepsilon$, for any $\varepsilon > 0$.
\end{theorem}

\begin{proof}
We consider the following game with $k=4$ types. Let $x \geq 1$ be an integer. There are $x$ red agents, one blue agent, one green agent, and one yellow agent; so, in total, there are $n=x+3$ agents. The graph is depicted in Figure \ref{fig:pos:sw}, where each node $p_i$, $i \in [x]$ is connected to nodes $q$, $s$ and $t$. 

First consider the assignment according to which the $x$ red agents occupy the nodes $p_1, \ldots, p_x$, the blue agent is at node $q$, the green agent is at node $s$, the yellow agent is at node $t$, and thus node $r$ is left empty. Observe that this assignment is not equilibrium since the blue agent has utility $1$ and can achieve a utility of $2$ by jumping to $r$. The optimal social welfare is at least the social welfare of this assignment, which is $3x+5$. Similarly, the optimal number of colorful edges is at least the number of colorful edges of this assignment, which is $3x+3$. 

We now argue about the social welfare and the number of colorful edges that can be achieves in any equilibrium. The possible assignments can be classified into one of the following three cases:
\begin{itemize}
\item
{\bf Case 1:} All $x$ red agents occupy nodes $p_1,p_2,\ldots,p_x$. The only possible equilibrium is when the blue, green, and yellow agents are in the triangle formed by $r$, $s$ and $t$. The social welfare of this equilibrium assignment is $2x+8$, and the number of colorful edges is $2x+3$. 

\item 
{\bf Case 2:} There is a red agent at $r$ and no red agents at $q$, $s$ and $t$. Then, all remaining $x-1$ red agents have to occupy $x-1$ nodes among $p_1,\ldots,p_x$. For an equilibrium, the three non-red agents must be at $r$, $s$ and the remaining node among $p_1,\ldots,p_x$. Hence, the social welfare of such an equilibrium is $2x+8$, and the number of colorful edges is $2x+3$. 

\noindent
Case 3: There is a red agent in at least one of $q,s$ and $t$. There must be at least $x-4$ red agents in $p_1,\ldots,p_x$. There can be at most $2(x-4)$ colorful edges out of the $3(x-4)$ edges to which these red agents are connected, since one of their 3 neighbors is occupied by a red agent. So, in this case, $\CE \le 2(x-4)+4(3)+4=2x+8$.

\item 
{\bf Case 3:} There is a red agent in at least one of $q$, $s$ and $t$. Then, there must be at least $x-4$ red agents at $p_1,\ldots,p_x$. Since one of their three neighbors is occupied by a red agent, the utility of these agent is at most $2$, and thus the social welfare of any such equilibrium is at most $2(x-4)(2)+7\cdot 3 = 2x+13$. In addition, the number of colorful edges is at most $2(x-4)+4\cdot3 + 4 = 2x+8$. 
\end{itemize}
Overall, the maximum possible social welfare of any equilibrium is at most $2x+13$ and the maximum possible number of colorful edges is $2x+8$. Consequently, the price of stability with respect to the social welfare is at least $\frac{3x+5}{2x+13}$ and the price of stability with respect to the number of colorful edges is at least $\frac{3x+8}{2x+8}$. Both quantities approach $3/2$ as $x$ approaches infinity. 
\end{proof}

\begin{figure}[t]
\centering

\tikzset{every picture/.style={line width=0.75pt}} 

\begin{tikzpicture}[x=0.75pt,y=0.75pt,yscale=-1,xscale=1]

\draw  [fill={rgb, 255:red, 255; green, 255; blue, 255 }  ,fill opacity=1 ] (301.71,522.4) .. controls (301.71,517.62) and (305.78,513.74) .. (310.8,513.74) .. controls (315.82,513.74) and (319.89,517.62) .. (319.89,522.4) .. controls (319.89,527.19) and (315.82,531.06) .. (310.8,531.06) .. controls (305.78,531.06) and (301.71,527.19) .. (301.71,522.4) -- cycle ;
\draw  [fill={rgb, 255:red, 255; green, 255; blue, 255 }  ,fill opacity=1 ] (302.81,560.21) .. controls (302.81,555.43) and (306.88,551.55) .. (311.9,551.55) .. controls (316.92,551.55) and (320.99,555.43) .. (320.99,560.21) .. controls (320.99,564.99) and (316.92,568.87) .. (311.9,568.87) .. controls (306.88,568.87) and (302.81,564.99) .. (302.81,560.21) -- cycle ;
\draw  [fill={rgb, 255:red, 255; green, 255; blue, 255 }  ,fill opacity=1 ] (263.32,478.11) .. controls (263.32,473.33) and (267.39,469.45) .. (272.41,469.45) .. controls (277.43,469.45) and (281.5,473.33) .. (281.5,478.11) .. controls (281.5,482.9) and (277.43,486.77) .. (272.41,486.77) .. controls (267.39,486.77) and (263.32,482.9) .. (263.32,478.11) -- cycle ;
\draw  [fill={rgb, 255:red, 255; green, 255; blue, 255 }  ,fill opacity=1 ] (342.29,478.11) .. controls (342.29,473.33) and (346.36,469.45) .. (351.38,469.45) .. controls (356.4,469.45) and (360.47,473.33) .. (360.47,478.11) .. controls (360.47,482.9) and (356.4,486.77) .. (351.38,486.77) .. controls (346.36,486.77) and (342.29,482.9) .. (342.29,478.11) -- cycle ;
\draw    (272.41,486.77) -- (301.71,522.4) ;
\draw    (351.38,486.77) -- (319.89,522.4) ;
\draw    (281.5,478.11) -- (342.29,478.11) ;
\draw  [fill={rgb, 255:red, 255; green, 255; blue, 255 }  ,fill opacity=1 ] (260.03,631.5) .. controls (260.03,626.72) and (264.1,622.84) .. (269.12,622.84) .. controls (274.14,622.84) and (278.21,626.72) .. (278.21,631.5) .. controls (278.21,636.29) and (274.14,640.16) .. (269.12,640.16) .. controls (264.1,640.16) and (260.03,636.29) .. (260.03,631.5) -- cycle ;
\draw  [fill={rgb, 255:red, 255; green, 255; blue, 255 }  ,fill opacity=1 ] (291.84,631.5) .. controls (291.84,626.72) and (295.91,622.84) .. (300.93,622.84) .. controls (305.95,622.84) and (310.02,626.72) .. (310.02,631.5) .. controls (310.02,636.29) and (305.95,640.16) .. (300.93,640.16) .. controls (295.91,640.16) and (291.84,636.29) .. (291.84,631.5) -- cycle ;
\draw  [fill={rgb, 255:red, 255; green, 255; blue, 255 }  ,fill opacity=1 ] (355.45,633.66) .. controls (355.45,628.88) and (359.52,625) .. (364.54,625) .. controls (369.56,625) and (373.63,628.88) .. (373.63,633.66) .. controls (373.63,638.45) and (369.56,642.32) .. (364.54,642.32) .. controls (359.52,642.32) and (355.45,638.45) .. (355.45,633.66) -- cycle ;
\draw   (319.81,632.82) .. controls (319.81,632.17) and (320.33,631.65) .. (320.96,631.65) .. controls (321.6,631.65) and (322.11,632.17) .. (322.11,632.82) .. controls (322.11,633.46) and (321.6,633.98) .. (320.96,633.98) .. controls (320.33,633.98) and (319.81,633.46) .. (319.81,632.82) -- cycle ;
\draw   (330.36,632.82) .. controls (330.36,632.17) and (330.88,631.65) .. (331.51,631.65) .. controls (332.15,631.65) and (332.66,632.17) .. (332.66,632.82) .. controls (332.66,633.46) and (332.15,633.98) .. (331.51,633.98) .. controls (330.88,633.98) and (330.36,633.46) .. (330.36,632.82) -- cycle ;
\draw   (341.54,633.08) .. controls (341.54,632.44) and (342.06,631.91) .. (342.69,631.91) .. controls (343.33,631.91) and (343.84,632.44) .. (343.84,633.08) .. controls (343.84,633.72) and (343.33,634.24) .. (342.69,634.24) .. controls (342.06,634.24) and (341.54,633.72) .. (341.54,633.08) -- cycle ;
\draw   (250,608.31) -- (386,608.31) -- (386,658) -- (250,658) -- cycle ;
\draw    (311.9,568.87) -- (312.52,608.31) ;
\draw    (310.8,531.06) -- (311.42,551.06) ;
\draw    (272.41,486.77) -- (312.52,608.31) ;
\draw    (351.38,486.77) -- (312.52,608.31) ;

\draw (259.76,626.48) node [anchor=north west][inner sep=0.75pt]  [font=\scriptsize]  {$\ p_{1}$};
\draw (292.67,626.48) node [anchor=north west][inner sep=0.75pt]  [font=\scriptsize]  {$\ p_{2}$};
\draw (355.18,628.64) node [anchor=north west][inner sep=0.75pt]  [font=\scriptsize]  {$\ p_{x}$};
\draw (304.54,554.02) node [anchor=north west][inner sep=0.75pt]  [font=\scriptsize]  {$\ q$};
\draw (304.54,517.29) node [anchor=north west][inner sep=0.75pt]  [font=\scriptsize]  {$\ r$};
\draw (266.15,473.01) node [anchor=north west][inner sep=0.75pt]  [font=\scriptsize]  {$\ s$};
\draw (345.12,474.09) node [anchor=north west][inner sep=0.75pt]  [font=\scriptsize]  {$\ t$};
\end{tikzpicture}
\caption{The graph of the game used to show a lower bound on the price of stability for both the social welfare and the number of colorful edges. Nodes $s$, $t$ and $q$ are connected to each of the nodes $p_1,\ldots,p_x$.}
\label{fig:pos:sw}
\end{figure}
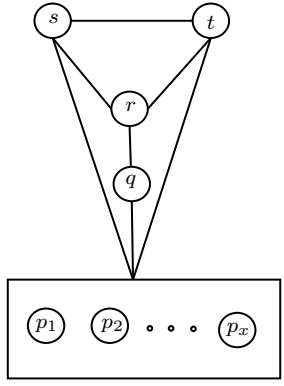

\section{Conclusion and Open Problems}
In this paper we considered a variety-seeking jump game in which agents occupy the nodes of a graph and aim to maximize the number of different types in their neighborhood. We showed the existence of equilibrium assignments for many classes of games depending on the number of agents and their types, as well as the structure of the graph. We also showed tight bounds on the price of anarchy with respect to the social welfare and the number of colorful edges, both of which measure in different ways the achieved diversity of equilibrium assignments.

Our work leaves open several interesting, yet quite challenging questions. While we have shown that equilibrium assignments exist for several classes of games, it remains unknown whether such assignments always exist or if there are games that do not admit any equilibrium. We conjecture that there is always at least one equilibrium and the game is potential when there are only two empty nodes (without any further restrictions on the structure of the underlying graph). Preliminary experiments using random graphs strongly support the latter claim; improving response cycles were found only when there were at least three empty nodes in the graph. A formal proof of this, however, remains elusive. 

In terms of measuring the diversity of equilibria, while we showed tight bounds on the price of anarchy in terms of the social welfare and the number of colorful edges, we were not able to show tight bounds on the price of stability. One could also consider many other objective functions to measure diversity, such as the degree of integration or variations of it \citep{schelling-journal}. 

\bibliographystyle{plainnat}
\bibliography{references}

\end{document}